\documentclass[sigconf]{acmart}
\usepackage{amsthm}
\usepackage{amsmath}
\usepackage[linesnumbered,ruled,vlined]{algorithm2e}
\usepackage{multirow}
\usepackage{graphicx}
\usepackage{xcolor}
\usepackage{colortbl}
\usepackage{tcolorbox}

\usepackage{arydshln} 

\usepackage{amsthm}
\theoremstyle{plain}  % 定义为"plain"风格, 使定理名称斜体
\newtheorem{theorem}{Theorem}
\usepackage{enumitem} % 提供灵活的列表设置
\newtheorem{remark}{Remark}
\usepackage{lineno}

\AtBeginDocument{%
  }

\copyrightyear{2026}
\acmYear{2026}
\setcopyright{cc}
\setcctype{by}
\acmConference[KDD 2026] {Proceedings of the 32nd ACM SIGKDD Conference on Knowledge Discovery and Data Mining V.1}{August 9--13, 2026}{Jeju Island, Republic of Korea.}
\acmBooktitle{Proceedings of the 32nd ACM SIGKDD Conference on Knowledge Discovery and Data Mining V.1 (KDD 2026), August 9--13, 2026, Jeju Island, Republic of Korea}
\acmISBN{979-8-4007-2258-5/2026/08}
\acmDOI{10.1145/3770854.3780249}
\begin{document}

\title[PSQE: A Theoretical-Practical Approach to Pseudo Seed Quality Enhancement for Unsupervised MMEA]{PSQE: A Theoretical-Practical Approach to Pseudo Seed Quality Enhancement for Unsupervised Multimodal Entity Alignment}%%PSQE: A Theory-Guided Framework for Precision and Coverage-Optimized Pseudo Seeds in Unsupervised Multimodal Entity Alignment;  % Theoretical and Practical Enhancements in Unsupervised Multimodal Entity Alignment via Pseudo Seed Refinement}  %
 % Towards Unsupervised Multimodal Entity Alignment: Enhancing pseudo seed Quality with Multimodal Information and Resampling Strategies


\author{Yunpeng Hong}
\email{hongyp@mail.hfut.edu.cn}
\orcid{0009-0000-6150-418X}
\affiliation{%
  \institution{Key Laboratory of Knowledge Engineering with Big Data (the Ministry of Education of China), Hefei University of
Technology}
  \city{Hefei}
  \state{Anhui}
  \country{CN}
}

\author{Chenyang Bu}
\email{chenyangbu@hfut.edu.cn}
\orcid{0000-0001-8203-0956}
\authornotemark[1]
\affiliation{%
  \institution{Key Laboratory of Knowledge Engineering with Big Data (the Ministry of Education of China), Hefei University of
Technology}
  \city{Hefei}
  \state{Anhui}
  \country{CN}
}

\author{Jie Zhang}
\email{jiezhang24@mail.hfut.edu.cn}
\orcid{0009-0005-4707-8272}
\affiliation{%
  \institution{Key Laboratory of Knowledge Engineering with Big Data (the Ministry of Education of China), Hefei University of
Technology}
  \city{Hefei}
  \state{Anhui}
  \country{CN}
}

\author{Yi He}
\email{yihe@wm.edu}
\orcid{0000-0002-5357-6623}
\affiliation{%
  \institution{Department of Data Science, College of William and Mary}
  \city{Williamsburg}
  \state{VA}
  \country{USA}
}

\author{Di Wu}
\email{wudi.cigit@gmail.com}
\orcid{0000-0002-7788-9202}
\affiliation{%
  \institution{College of Computer and Information Science, Southwest University}
  \city{Chongqing}
  \country{CN}
}

\author{Xindong Wu}
\authornote{Corresponding authors.}
\email{xwu@hfut.edu.cn}
\orcid{/0000-0003-2396-1704}
\affiliation{%
  \institution{Key Laboratory of Knowledge Engineering with Big Data (the Ministry of Education of China), Hefei University of
Technology}
  \city{Hefei}
  \state{Anhui}
  \country{CN}
}

\renewcommand{\shortauthors}{Yunpeng Hong et al.}

\begin{abstract}
Multimodal Entity Alignment (MMEA) aims to identify equivalent entities across different data modalities, enabling structural data integration that in turn improves the performance of various large language model applications.
To lift the requirement of labeled seed pairs that 
are difficult to obtain, recent methods shifted to an
unsupervised paradigm using pseudo-alignment seeds.
However, unsupervised entity alignment in multimodal settings remains underexplored, mainly because the incorporation of multimodal information often results in imbalanced coverage of pseudo-seeds within the knowledge graph.
To overcome this, we propose PSQE (\textbf{P}seudo-\textbf{S}eed \textbf{Q}uality \textbf{E}nhancement) to improve the precision and graph coverage balance of pseudo seeds via multimodal information and clustering-resampling. 
Theoretical analysis reveals the impact of pseudo seeds on existing contrastive learning-based MMEA models. 
In particular, pseudo seeds can influence the attraction and the repulsion terms in contrastive learning at once, 
whereas imbalanced graph coverage causes models to prioritize high-density regions, thereby weakening their learning capability for entities in sparse regions.
Experimental results validate our theoretical findings and show that PSQE as a plug-and-play module can improve the performance of baselines by considerable margins.
\end{abstract}



% <concept_id>10002951.10002952.10003219.10003221</concept_id>
% <concept_desc>Information systems~Wrappers (data mining)</concept_desc>

% \ccsdesc[500]{Information systems~Wrappers (data mining)}

% %% A "teaser" visual appears between the author and affiliation
% %% information and the body of the document, and typically spans the


%%
%% This command processes the author affiliation and title
\maketitle
\newcommand\kddavailabilityurl{https://doi.org/10.1145/3770854.3780249}
\ifdefempty{\kddavailabilityurl}{}{
\begingroup\small\noindent\raggedright\textbf{Resource Availability:}\\
% please change the following context to include multiple artifacts if necessary, including data, models, code, etc.
The source code of this paper has been made publicly available at \url{https://github.com/flyfish259/PSQE}.
\endgroup
}
\begin{sloppypar}
%\vspace{-1em}
\section{Introduction}
Multimodal entity alignment (MMEA) is pivotal for integrating heterogeneous multimodal data (e.g., text, images, videos) \cite{ektefaie2023multimodal,liu2024cross} across diverse sources, enabling the identification of equivalent entities and breaking data silos. This capability is essential for enhancing knowledge-driven applications~\cite{yu2020cross,ektefaie2023multimodal, retrieval,zhuo2025effective,DBLP:conf/ijcai/ZhuoPWW0LW025}, ranging from improving the GraphRAG frameworks~\cite{edge2024local,bu2025query,huang2025mitigate} to augmenting large language models (LLMs) with multimodal context~\cite{kasneci2023chatgpt}, to name a few.

A prominent challenge in scaling MMEA to real-world scenarios lies in the reliance on labeled training data. Most existing MMEA methods~\cite{MSNEA, ack-mmea} rely on supervised learning, where a model is trained using expert-labeled seed pairs to compute the similarity between entities in a graph and determine alignment relationships. 
To wit, 30\% of known seed pairs are used as labels for MMEA training in the benchmark DBP15K dataset~\cite{sun2017cross}.
The manual annotation process is prohibitively expensive. 
Although one may think using 
large language models (LLMs) for the MMEA task~\cite{tanwar-etal-2023-multilingual,10948482}, they mostly suffer from sizable computational overheads. 
For instance, knowledge graphs having hundreds of thousands entities per language (i.e., each language as one data modality) are a norm, 
aligning all entities across languages would incur prohibitive costs.
These limitations agitate growing interest in unsupervised MMEA methods that can operate without ground-true seed-pair labels.

%While effective, labeling alignment seeds manually is time-consuming and expensive. This limitation underscores the growing significance of unsupervised MMEA methods in practical applications.%However, manually labeling alignment seeds in practical applications is time-consuming and expensive, making unsupervised MMEA methods highly significant from a practical standpoint.
%Despite recent advances, a critical challenge persists when scaling MMEA to real-world scenarios: \textbf{{How can we achieve accurate alignment without relying on costly labeled data}}?

\begin{figure}[!t]
\centering
\includegraphics[width=\linewidth]{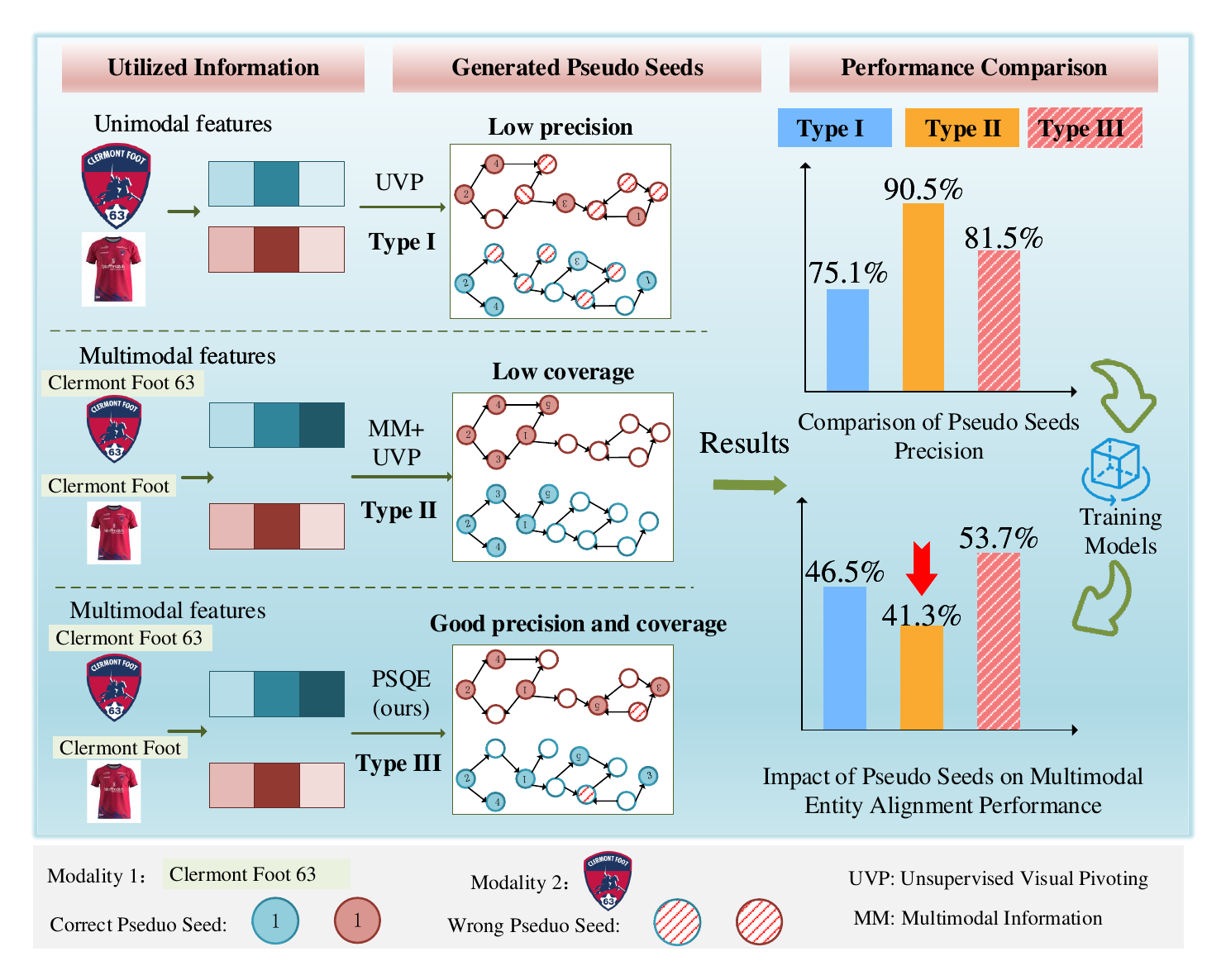}
\caption{Comparative analysis of three types of pseudo seed generation on the FR-EN dataset for Multimodal Entity Alignment (MMEA): (1) Type I: Unimodal, (2) Type II: Multimodal, and (3) Type III: Distribution-Aware Multimodal. An interesting phenomenon is observed: Type II achieves higher precision than Type I but performs worse in the downstream MMEA task. In contrast, Type III, due to its better distribution balance, yields the best performance. }
\label{motivate}
\vspace{-1em}
\end{figure}

Alas, the current studies on unsupervised MMEA remains limited,
as their designs are sensitive to 
the quality of automatically generated pseudo-aligned seeds.
In particular, no unsupervised MMEA method has considered 
balancing both {\em precision} and {\em distribution coverage}
at once. 
Consider an example illustrated in 
Fig.~\ref{motivate}, we observe that when comparing pseudo seed generation strategies: 
Type II seeds, which are based on the mechanism of Type I~\cite{liu2021visual} while integrating multi-modal information, achieve 90.5\% precision but yield worse alignment than Type I (75.1\% precision) due to imbalanced graph coverage, and Type III demonstrates optimal performance by jointly optimizing both metrics. This raises \textbf{two research questions}: 
\begin{itemize}[left=0pt, label=•, itemsep=0.2em]
    \item {What factors govern the impact of pseudo seeds on MMEA}? 
    \item {How to generate high-quality seeds without supervision}?
\end{itemize}

In this paper, we explore the two questions and provide 
answers by proposing a \textbf{P}seudo-\textbf{S}eed \textbf{Q}uality \textbf{E}nhancement (PSQE) framework for unsupervised MMEA. 
The design of PSQE is motivated by our analysis 
on how pseudo seeds impact 
contrastive learning (CL)-based MMEA models (Theorem~\ref{theorem1}).
Specifically, the lower bound of 
the intra-modal CL loss can be decomposed into two terms, namely 
i) an attraction term that minimizes the distance between aligned entity pairs  and 
ii) a repulsion term that maximizes 
separation between negative samples.
We further observe that the attraction term is 
governed by seed precision, where mismatched seeds will introduce biased gradients, pushing correct pairs apart (Fig. \ref{figtheroem1}).
Likewise, the repulsion term is 
determined by graph coverage balance, where an imbalanced coverage will skew gradients toward dense regions, under-optimizing sparse entities (Fig. \ref{figtheorem2}). 
To optimize the balance between two terms, 
PSQE integrates multimodal signals with a novel resampling strategy across three stages to simultaneously boost seed precision and balance graph coverage. 
As such, multimodal precision strengthens the attraction term, while clustering-resampling controls repulsion.
As shown in Fig. \ref{motivate}, 
seeds generated by our PSQE (Type III) excel. 
%

%

%

% Theoretically, we dissect how pseudo seeds influence contrastive learning-based MMEA models through the lens of Intra-modal Contrastive Loss (ICL). As proven in Theorem 1, the ICL lower bound decomposes into two key terms: (1) Attraction term: Governed by seed precision, it minimizes the distance between aligned entity pairs (Fig. \ref{figtheroem1}). Mismatched seeds introduce biased gradients, pushing correct pairs apart. (2) Repulsion term: Determined by graph coverage balance, it maximizes separation between negative samples. Imbalanced coverage skews gradients toward dense regions, under-optimizing sparse entities (Fig. \ref{figtheorem2}). 

Experiments on two large-scale benchmarks, DBP15k (cross-lingual) and DWY15k (monolingual), demonstrate that PSQE significantly improves the performance of state-of-the-art unsupervised MMEA methods. For instance, when integrated with MEAformer, PSQE improves the Hits@1 by 3.8\% (ZH-EN), 2.0\% (JA-EN), and 1.4\% (FR-EN), while maintaining robustness across varying seed initialization settings. Our ablation studies further confirm the necessity of multimodal fusion (where visual features contribute to a gain of >10\% over MEAformer on DBP15K) and balanced graph coverage (resulting in an MRR improvement of 1.1\%). 

% These results substantiate that  PSQE can bridge the gap between unsupervised seed generation and supervised performance, offering a scalable solution for real-world MMEA applications. 

%which leverages multimodal information and resampling strategies across three stages to improve the precision and graph coverage balance of pseudo seeds, thereby enhancing their overall quality. Additionally, we provide a theoretical analysis to reveal the impact of pseudo-aligned seeds on current popular contrastive learning-based MMEA models. We find that pseudo-aligned seeds influence both the attraction and the repulsion terms in contrastive learning: the former depends on the precision of the pseudo seeds, while the latter is affected by the balance of graph coverage. When graph coverage is imbalanced, gradient updates tend to favor entities in high-density areas, which negatively impacts the learning of entities in sparse regions.

\textbf{Specific contributions} of this paper includes:
\begin{itemize}[left=0pt, label=•, topsep=0pt, itemsep=0.2em]
     \item We propose PSQE, 
     a plug-and-play module for unsupervised MMEA. 
     It is the first framework to jointly optimize pseudo seed precision and coverage distribution, of which the performance is on a par with its supervised competitors. 
    \item  A theoretical analysis is conducted on the intra-model contrastive learning dynamics, giving explanation on how the seed quality impacts alignment performance.
    \item Extensive experiments substantiate that PSQE can significantly improve performance of the state-of-the-art models, validating the importance of seed quality in unsupervised MMEA.
\end{itemize}

% To further investigate the impact of pseudo seeds on contrastive learning-based MMEA models, we theoretically reveal the learning process of such models. We find that pseudo-alignment seeds influence both the attraction and the repulsion terms in contrastive learning, and both terms are affected by the precision and graph coverage balance of the pseudo seeds. Consequently, the performance of the model is influenced not only by the precision of pseudo seeds but also by their graph coverage balance. Based on this finding, we propose a \textbf{P}seudo-\textbf{S}eed \textbf{Q}uality \textbf{E}nhancement (PSQE) method that uses multimodal information and resampling strategies. Specifically, in the first stage, we enhance precision through multimodal features and optimize the graph coverage balance of pseudo seeds by using clustering methods. In the second stage, we use feature resampling with pseudo seeds to optimize their graph coverage balance and correct errors using multimodal entity features to improve their precision. In the third stage, we expand the pseudo seeds by resampling neighboring entities, optimizing their graph coverage balance, and applying the error correction from the second stage to enhance the quality of the final pseudo seeds. On two large-scale benchmarks, DBP15k for cross-lingual EA and DWY15k for monolingual EA, our PSQE method performs very well. Overall, our contributions are as follows:
% \vspace{-10pt}

%\vspace{-1em}
\section{Related Work}
\subsection{Multimodal Entity Alignment}
Multimodal entity alignment (MMEA) refers to the process of identifying and matching equivalent entities across different knowledge graphs using various modalities of information (such as text, images, etc.). Most existing MMEA models adopt a supervised approach, where different modality information is processed and fused to generate the final entity vector representation~\cite{zhu2023mmiea,bu2024automatic, ack-mmea,ni2023psnea, MSNEA,zhang2025robust}.
Since supervised methods rely on large-scale, high-quality annotated data, which is costly and limited in applicability, unsupervised MMEA~\cite{liu2021visual} methods have gained widespread attention.

In the field of \textbf{unsupervised unimodal entity alignment}~\cite{he2019unsupervised,zeng2021towards,tang2023fused,liu2022selfkg,wang2024pseudo,iclea}, research primarily relies on unimodal information (such as names or text) to mine pseudo-labels~\cite{meng2025se}. Current unsupervised seed acquisition methods can be categorized into three types: 1) \textbf{Symbolic-based methods} typically rely on existing symbolic knowledge bases, rules, or dictionaries to obtain seeds. For instance, SE-UEA~\cite{jiang2023integrating} uses entity symbolic similarity to assist in seed acquisition. 2) \textbf{Pre-trained model-based methods} primarily use large pre-trained models such as BERT~\cite{devlin-etal-2019-bert}, LaBSE~\cite{feng-etal-2022-language}, etc., to extract entity features for seed generation. For example, FGWEA~\cite{tang2023fused} uses LaBSE to obtain vector representations of entity names, relationships, and attributes and acquires seeds through similarity comparison. 3) \textbf{LLMs-based methods} leverage the inherent knowledge of large models to automatically learn the potential relationships between entities. For example, LLM4EA~\cite{chen2024entity} generates alignment seeds by re-ranking candidate entity pairs.

Unsupervised multimodal entity alignment research remains relatively scarce when extending to multimodal scenarios. Current methods typically generate pseudo seeds first and then learn a contrastive learning model for alignment. EVA~\cite{liu2021visual}, as the first unsupervised multimodal entity alignment method, mainly utilizes the Unsupervised Visual Pivoting (UVP) technique to generate pseudo-alignment seeds using image information and then trains the model based on Neighborhood Component Analysis (NCA) loss. MEAformer~\cite{chen2023meaformer} further expands on this idea using image and name information to generate pseudo seeds and then trains the model based on Intra-modal Contrastive Loss (ICL).
Although some existing methods have made progress in unsupervised multimodal entity alignment, current unsupervised multimodal methods still rely on single modality information to generate pseudo-alignment seeds.

\vspace{-0.6em}
\subsection{Balanced Data Distribution for Contrastive Learning}
Due to the imbalance in the distribution of labeled data, traditional contrastive learning methods~\cite{cui2021parametric,kang2020exploring} can produce bias in the model training process. To address this problem, existing research has tried to mitigate the effects of distributional imbalance by focusing on both the data level and the model level.
\textbf{At the data level,} the imbalance problem is usually mitigated by optimizing data sampling. For example, MAK~\cite{mak} effectively improves the performance of contrastive learning on imbalanced seed data by strategically selecting unlabeled data from external data sources.
% KCL~\cite{kang2020exploring}, on the other hand, adopts a two-stage learning paradigm and uses the same number of positive samples for all categories in each batch.
% PaCo~\cite{cui2021parametric} effectively overcomes the performance degradation of self-supervised contrastive learning on unbalanced data by introducing a set of category-based learnable centers.
\textbf{At the model level}, the data imbalance problem is mainly tackled by optimizing the loss function.
ImGCL~\cite{zeng2023imgcl} better maintains the intrinsic structure of the graph by weighting the important nodes in the graph.
Bacon~\cite{zhu2022balanced} proposes an adaptive mechanism that dynamically adjusts the learning rate of categories according to the imbalanced distribution of labels, according to the distribution of labels to dynamically change the learning rate of categories, thus effectively mitigating the negative impact of imbalance.

However, existing methods mainly target the classification problem, while our task is a one-to-one graphical entity matching problem, so these methods cannot effectively solve the unsupervised multimodal entity alignment problem. 
% Our study reveals the potential causes of the pseudo-label graph coverage balance imbalance problem in this context through a theoretical analysis of comparative learning.

\begin{figure*}[htbp]
\centering
\includegraphics[width=\linewidth]{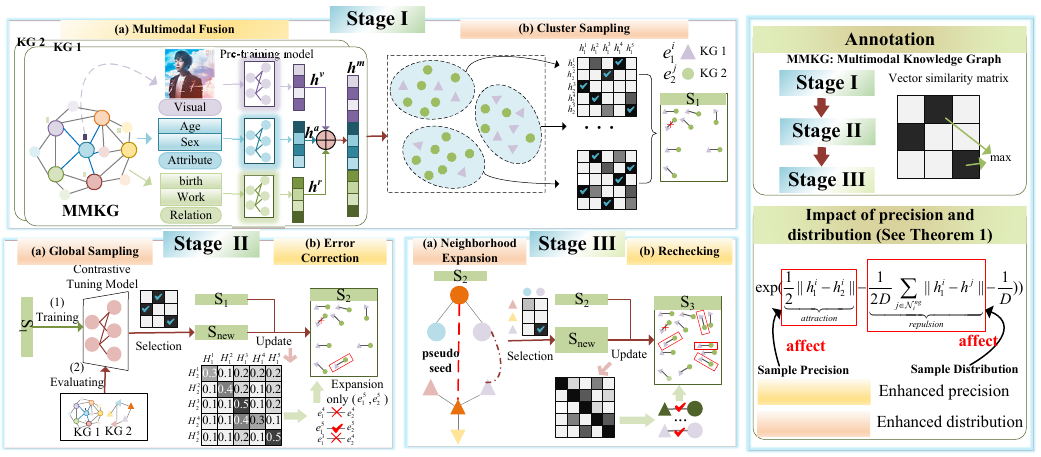}
   \caption{Overall framework of PSQE. 
 PSQE optimizes the precision and the graph coverage balance (distribution) of pseudo-alignment seeds in three stages to enhance their quality.}
\label{frame}
\vspace{-1em}
\end{figure*}

\vspace{-0.5em}
\section{Preliminaries}
\subsection{Multimodal Entity Alignment}
We define a $MMKG$ as $G = (E, R, A, V)$, where each ${\rm{e}} \in E$, ${\rm{r}} \in R$, ${\rm{a}} \in A$, and ${\rm{v}} \in V$ denotes an entity, a relation, an attribute, and a visual respectively.
Given two multi-modal knowledge graphs $G_1 = (E_1, R_1, A_1, V_1)$ and $G_2 = (E_2, R_2, A_2, V_2)$, the task of multimodal entity alignment targets to discover the set of equivalent entity pairs between $G_1$ and $G_2$, denoted as $\mathcal{M}=\{ (e_1^i,e_2^j)| e_1^i \equiv e_2^j, e_1^i \in E_1, e_2^j \in E_2\}$, where $e_1^i \equiv e_2^j$ means an equivalence entity between $e_1^i$ and $e_2^j$.

In the unsupervised setting, the MMEA model typically predicts 
$\mathcal{M}$ without relying on training data. The existing methods typically leverage pre-trained models to extract features from different modalities, and then use these features to identify pseudo-aligned entity seeds $\mathcal{S}$, where $\mathcal{S} = \{(e_1^i, e_2^j)| e_1^i \approx e_2^j\}$, and these seeds are subsequently used to train the MMEA model.

%\vspace{-1.2em}
\subsection{Contrastive Learning Loss}
In training multimodal entity alignment models, Intra-modal Contrastive Loss (ICL) is commonly used~\cite{ni2023psnea, mclea,chen2023meaformer} to clarify the boundaries of entity embeddings.

Specifically, ICL constructs negative samples aligning seed $\mathcal{S}$ and following the assumption of one-to-one alignment constraint. For two knowledge graphs $G1$ and $G2$, positive samples $P=\{ (e_1^i,e_2^j)| e_1^i \equiv e_2^j, e_1^i \in E_1, e_2^j \in E_2\}$ refer to the entity pairs in pseudo seed $\mathcal{S}$, and negative samples refer to ${\mathcal N}_i^{ng} = \left\{ {e_1^j|\forall e_1^j \in {E_1},j \ne i} \right\} \cup \left\{ {e_2^j|\forall e_2^j \in {E_2},j \ne i} \right\}$, where $E$ represents a batch of pseudo-aligned entity seeds. Furthermore, ICL constrains the embedding space with an in-batch negative sampling strategy, keeping semantically similar entities from the same knowledge graph close and improving sampling efficiency. The alignment probability distribution~\cite{mclea} can be defined as:
\begin{equation}
\scriptsize
\label{pm}
\begin{aligned}
p(h_1^i,h_2^j) = \frac{{\zeta (h_1^i,h_2^j)}}{{\zeta (h_1^i,h_2^j) + \sum\limits_{{e^j} \in {\mathcal N}_i^{ng}} \zeta  (h_1^i,{h^j})}}
\end{aligned}
\end{equation}
where ${\zeta_m}(h_1^i,h_2^j) = exp({h_1^i}^{\top } \cdot h_2^j/\tau )$, $\tau $ is a hyperparameter and $h^i_1$ is the embedding of entity $e^i_1$. Since entity alignment has direction, the distribution of Eq.~\ref{pm} is asymmetric and directional for each input. Therefore, a bi-directional alignment objective is used for each modality, as
\begin{equation}
\label{icl}
{L}^{\mathrm{ICL}}=-log(\frac{1}{2}(p(h_1^i,h_2^j)+p(h_2^j,h_1^i)))
\end{equation}

ICL loss forces the input embeddings to respect the similarity of the entities in the original embedding space and can distinguish between embeddings of the same entity from others in different knowledge graphs~\cite{mclea,wu2021knowledge}.

\section{Method}

\begin{itemize}[left=0pt, label=•, itemsep=0.5em]
    \item \textbf{O1: How to improve the precision of pseudo-seeds?} 
    % Corresponds to sections 4.2.1 (Fig. 4 Stage \uppercase\expandafter{\romannumeral 1} (a)), 4.3.2 (Fig. 4 Stage \uppercase\expandafter{\romannumeral 2} (b), \uppercase\expandafter{\romannumeral 3} (b)).
    \item \textbf{O2: How to optimize the coverage distribution of pseudo-seeds?}
    % Corresponds to sections 4.2.2 (Fig. 4 Stage \uppercase\expandafter{\romannumeral 1} (b)), 4.3.2 (Fig. 4 Stage \uppercase\expandafter{\romannumeral 2} (a)), 4.4 (Fig. 4 Stage \uppercase\expandafter{\romannumeral 3} (a)).
\end{itemize}

To address O1 and O2, the PSQE framework enhances the precision and coverage of pseudo-seeds through a three-stage strategy.

\textbf{In terms of precision optimization:} Stage \uppercase\expandafter{\romannumeral 1} integrates multimodal information to reduce the bias of single modality through complementary modal information, laying the foundation for high-quality seed selection; Stage \uppercase\expandafter{\romannumeral 2} introduces contrastive learning fine-tuning and error correction mechanisms to fine-tune the embedding representations of entities and eliminate erroneous seeds that cause embedding conflicts; Stage \uppercase\expandafter{\romannumeral 3} ensures the precision of newly added seeds through neighborhood structure alignment constraints and secondary error correction verification. 

\textbf{In terms of coverage distribution optimization:} Stage \uppercase\expandafter{\romannumeral 1} implements regional sampling based on semantic clustering, forcing the model to cover representative entities of different clustering centers; Stage \uppercase\expandafter{\romannumeral 2} captures cross-clustering alignment seeds that may arise through global sampling, enriching the types of pseudo-seeds from a global perspective; Stage \uppercase\expandafter{\romannumeral 3} employs a neighborhood expansion strategy to propagate entities through graph structure, filling the coverage gaps in low-density areas and supplementing sparse entities.

% \subsection{Framework}
% As shown in Fig.~\ref{frame},
% PSQE aims to improve the quality of multimodal entity-aligned pseudo seeds by improving the distribution and precision of pseudo seeds through three stages of optimization. To answer objectives O1 and O2, in the first stage, we improve the precision through multimodal features and optimize the distribution of pseudo seeds by combining clustering methods~\cite{hamerly2003learning}. In the second stage, the pseudo seeds from the first stage are used to expand the pseudo seeds using entity feature reconstruction methods to optimize their distribution further, and at the same time, error correction is carried out with the help of multimodal entity features to remove some of the pseudo seeds to improve the precision. In the third stage, we resample some of the pseudo seeds by using the neighborhood entities to expand the pseudo seeds to optimize their distribution and apply the error correction method in the second stage to improve the precision of the final pseudo seed set.

\subsection{Stage \uppercase\expandafter{\romannumeral 1}: Multimodal Fusion \& Cluster Sampling}
In Stage \uppercase\expandafter{\romannumeral 1}, to improve the precision of pseudo-seeds, we represent entity features by integrating multimodal features, thereby reducing the bias caused by single-modal entity features. To optimize the coverage distribution, we cluster the entity information from two knowledge graphs and force the generation of pseudo-seeds in different semantic regions, thereby distributing the seeds.

\subsubsection{Multimodal Fusion}
Given that multimodal features can enhance the representation of an entity, we first encode its visual, relational, and attribute modalities and stitch them together to form a complete representation of the features of the entity. For the visual modality, we adopt a representative pre-trained visual model, denoted as $PreM$, as the encoder. We extract the output from the final layer before the logits to obtain the visual embedding $h^{v_i}$ for each visual instance $v_i$ associated with the entity $e^i$. In practice, following common settings in prior works~\cite{mclea,liu2021visual,chen2023meaformer}, we employ ResNet as the pre-trained encoder.
We use BERT for the semantic embedding of relationship and attribute modalities of entities. Since an entity may have multiple attributes $a_i$ and relations $r_i$, we average the semantic embeddings for multiple attributes and relations.
\begin{equation}
    \left\{ {\begin{array}{*{20}{c}}
{{h^{v_i}} = PreM({v^i})}\\
{{h^{a_i}} = AVE\left( {BERT({a_i[1]},{a_i[2]},...} \right)}\\
{{h^{r_i}} = AVE\left( {BERT({r_i[1]},{r_i[2]},...} \right)}
\end{array}} \right.
\end{equation}
The final entity representation $h^i$ can be obtained by stitching together these modal information:
\begin{equation}
    {h^i} = {h^{v_i}} \oplus {h^{a_i}} \oplus {h^{r_i}}
\end{equation}
where $\oplus$ denotes concatenation.

\subsubsection{Cluster Sampling}
To prevent the pseudo-alignment seeds from becoming overly concentrated in the knowledge graph, which could result in insufficient learning of more sparse entities, we employ the K-means clustering algorithm~\cite{hamerly2003learning} to partition both knowledge graphs into several sub-blocks. Then, within each sub-block, pseudo seed pairs are generated to explore more potential alignments, thereby promoting a more even graph coverage balance of the pseudo seeds.

Specifically, there are two knowledge graphs $G_1$ and $G_2$, which contain the sets of entities $\{ e_1^1, e_1^2, \dots\}$ and $\{ e_2^1, e_2^2, \dots \}$, where $e_1^i$ and $e_2^i$ denote the entities from graphs $G_1$ and $G_2$, respectively. To cluster the entities in the knowledge graph, we use the K-means clustering method, which divides the entities of the two atlases into $K$ clusters, denoted as ${ C_1, C_2, \dots, C_K }$, where $C_k$ denotes the $k$th cluster containing all the entities classified into that cluster.

In each cluster $C_k$, we compute its similarity $\text{Sim}(h_1^{i}, h_2^{j})$ for the pairs of entities $(e^i_1, e^j_2)$ from the graphs $G_1$ and $G_2$ within the cluster, where $e^i_1, e^j_2 \in C_j$.
Next, we sort all entity pairs based on the similarity values from high to low. To further select the pseudo-alignment seeds, we pick the top $m_j = \frac{\text{len}(C_j) \cdot n}{\text{len}(G_1) + \text{len}(G_2)}$ pairs of entities with the highest similarity values $(e_1^i, e_2^j)$ from each cluster $C_j$, where $n$ denotes the total number of pseudo seeds required. In particular, the entities in the selected entity pairs must satisfy the condition that they are not yet included in the current set of pseudo seeds $S$.
Then, the pseudo seed $S1$ of stage \uppercase\expandafter{\romannumeral 1} is obtained.
% See Appendix~\ref{algorim} for the specific algorithmic process.

\subsection{Stage \uppercase\expandafter{\romannumeral 2}: Global Sampling \& Error Correction}
% In the second stage, to achieve O1 and O2, we use contrastive learning to reconstruct the vector representation of entities to mine more pseudo seeds to expand the coverage in the knowledge graph. At the same time, the high accuracy property of multimodal features is exploited for error correction to compensate for the possible accuracy degradation caused by the optimized graph coverage.

In Stage \uppercase\expandafter{\romannumeral 2}, to optimize the seed coverage distribution, we first fine-tune the entity embeddings based on contrastive learning to make them more accurately reflect the semantic information of the entities. Then, we perform global sampling in the knowledge graph instead of sampling within clusters, which allows us to capture cross-cluster aligned seeds and expand the coverage of seeds. Meanwhile, to optimize the precision of seeds, we conduct error-checking on the expanded pseudo-seeds after fine-tuning the entity representations, identifying and removing pseudo-seed pairs that may conflict with those from Stage \uppercase\expandafter{\romannumeral 1}, thereby improving the overall quality of the seeds.

\subsubsection{Global Sampling} To enhance entity features and achieve resampling through pseudo seed expansion, we propose a feature enhancement method based on contrastive tuning. Specifically, we map and reconstruct the visual features, attribute features, and relationship features of entities independently to enhance the expressiveness of the features. For the visual feature $h^v$ of entities, it is first mapped by an independent linear layer, and then the neighborhood information is aggregated using Graph Attention Network (GAT)~\cite{gat} to enhance the expressiveness of the visual features. For the attributes and relations of entities, we use an independent fully connected layer to transform the attribute features and relation features of entities respectively, and get the new entity attribute features and relation features as follows:
\begin{equation}
    \left\{ {\begin{array}{*{20}{c}}
{{H^{v_i}} = GAT(linear({h^{v_i}}))}\\
{{H^{a_i}} = linear({h^{a_i}})}\\
{{H^{r_i}} = linear({h^{r_i}})}
\end{array}} \right.
\end{equation}

In addition, we train the model using ICL loss:
\begin{equation} loss = ICL(H_1^{v},H_2^{v}) + ICL(H_1^a,H_2^a) + ICL(H_1^r,H_2^r) \end{equation}
where $H_1^v$ and $H_2^v$ denote the entity vectors in the pseudo seed sets of the two knowledge graphs.

After obtaining the new entity features, we use these features to expand the pseudo seeds. The steps are as follows: First, the individual modal vectors of entity \( e_1^i \) are spliced to obtain the new entity feature vector \( H_1^{i} \).
  % : \[H_1^{{e^p}} = H_1^v \oplus H_1^a \oplus H_1^r\]
Next, the similarity between all pairs of entities in the two knowledge graphs is calculated. 
Based on similarity scores, the entity pairs that have not yet been added to the pseudo seed set are selected, and the \( n \) entity pairs with the highest similarity are selected as the new pseudo seeds ${S}_{new}$. 

\begin{figure}[!t]
\centering
\includegraphics[width=\linewidth]{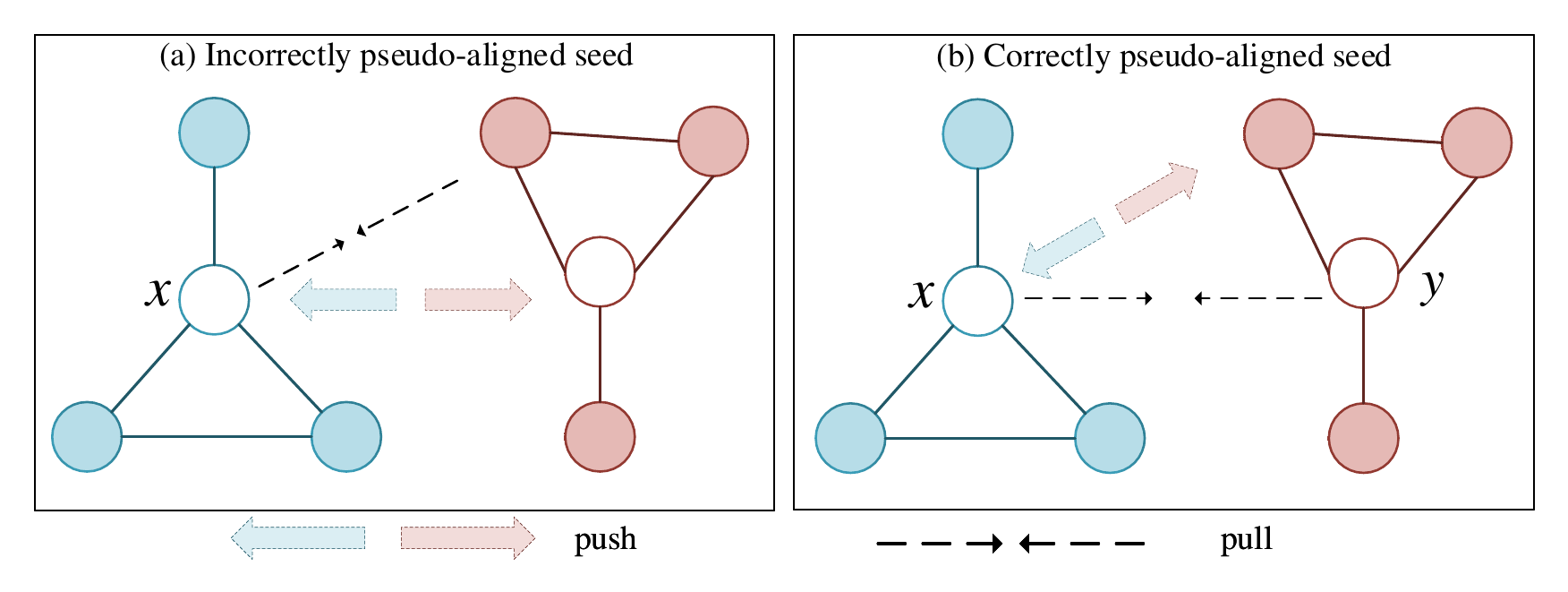}
\caption{Example of an attraction term for Remark 1: \( x \) and \( y \) are correct entity pairs.
Wrongly aligned seeds can push away vector representations from the correct entity, resulting in variability from the vector of the correct entity.
}
\label{figtheroem1}
\end{figure}

\subsubsection{Multimodal Information Correction}
After mining pseudo seeds, we propose a correction mechanism based on original multimodal features to improve the precision of the pseudo seed set.

Let $h_1 = [h_1^{1}, h_1^{2}, \dots, h_1^{n}]$, $h_2 = [h_2^{1}, h_2^{2}, \dots, h_2^{n}]$ denote the feature matrices of the entities in the pseudo seeded set $S_{new}$, respectively. To detect anomalies in the pseudo seeds, we compute the similarity matrix:
$M = h_1^{T} \cdot h_2$.
For each row \( k \) in the similarity matrix, we check the following condition:
$\text{M}_{kk} \neq \max_{l} \text{M}_{kl}$.
That is, we verify whether the diagonal element \( \text{M}_{kk} \) is the largest value in row \( k \). If the condition is satisfied, the corresponding pseudo seed \((e_1^k, e_2^k), (e_1^k, e_2^l)\) is deemed questionable. Such pseudo seeds are then removed from the pseudo seed set.
After this screening process, the updated pseudo seed set is $S2$. 

\subsection{Stage \uppercase\expandafter{\romannumeral 3}: Neighborhood Expansion  \& Rechecking}
In Stage \uppercase\expandafter{\romannumeral 3}, to optimize the coverage distribution, we conduct local area sampling through neighborhood expansion, thereby supplementing sparse entities and making the sample distribution more balanced. After optimizing the distribution, precision usually declines. To enhance precision, we employ the error correction mechanism from Stage \uppercase\expandafter{\romannumeral 2  } for secondary error correction.

\subsubsection{Neighborhood Expansion}
Since the neighboring entities of aligned entities in pseudo-alignment seed pairs are usually semantically or relationally similar~\cite{zhu2017iterative}, we speculate that their neighboring entities can also form alignment relationships. Therefore, we increase the number of pseudo-alignment seeds by expanding the pseudo-alignment seed pairs of neighboring entities and, at the same time, supplementing the information of the knowledge graph around the scattered entities to enhance the attention of the originally scattered entity part in the gradient update.
Meanwhile, combined with the error correction strategy in Stage \uppercase\expandafter{\romannumeral 2}, it can effectively alleviate the problem of degradation of pseudo seed precision that may occur during the resampling process. 

Specifically, for each pseudo seed pair $(e_1^i, e_2^i)$, we compute the similarity of the entity representation of its neighbors:
\begin{equation}
\small
    \text{Sim}(i', j') = {(h_1^{i'} \oplus H_1^{i'})^T}(h_2^{j'} \oplus H_2^{j'}) \quad \text{for} \quad i' \in \mathcal{N}(e_1^i), j' \in \mathcal{N}(e_2^i)
\end{equation}
where $\mathcal{N}(e_1^i)$ and $\mathcal{N}(e_2^i)$ denote the sets of neighboring entities of entities $e_1^i$ and $e_2^i$ respectively.
Next, we select entity pairs $(i', j')$ from the similarity computation results and filter the entity pairs that satisfy the following conditions: similarity higher than the safety threshold $\eta$, entities in the entity pairs do not appear in the pseudo seed set $S_2$.

\begin{figure}[!t]
\centering
\includegraphics[width=\linewidth]{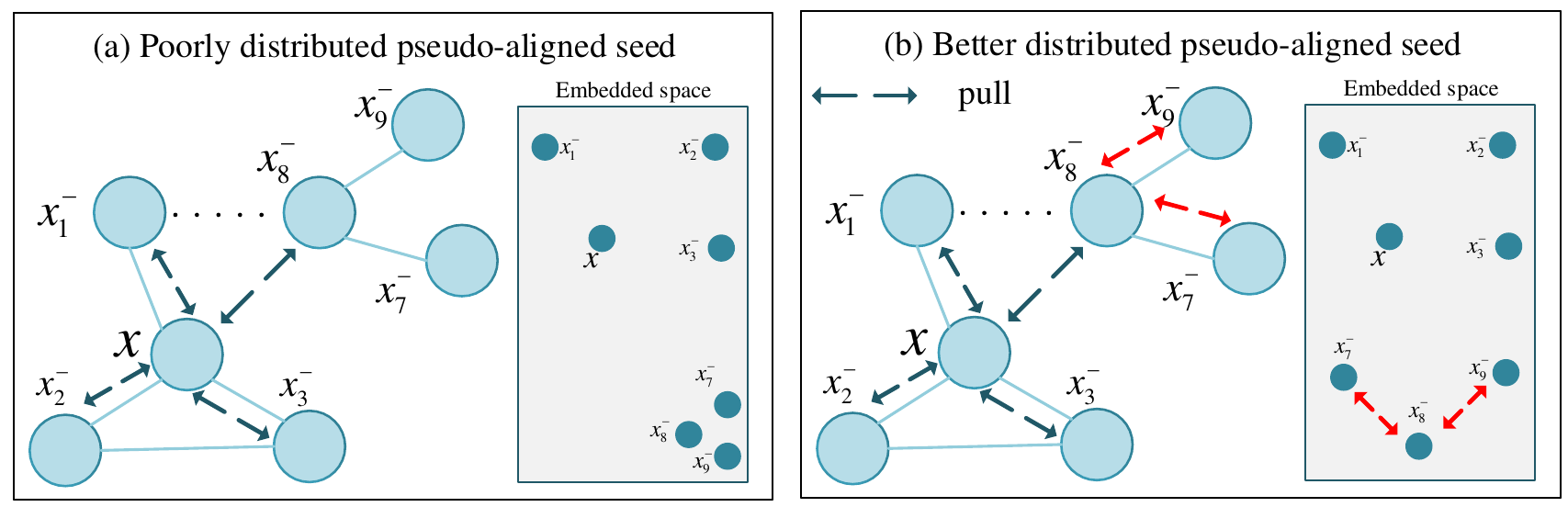}
\caption{Examples of a repulsion term of Remark 2: $x^-$ is a neighboring entity, and red arrows indicate new negative samples. Negative samples were not added to $x_7^-$, $x_9^-$ in the left figure.
Unbalanced graph coverage leads to focusing on $x_1^-$, $x_2^-$, and $x_3^-$ during optimization and under-optimization of entities in the region around $x_8^-$, which are difficult to distinguish, e.g., the bottom right corner of Fig. (a).
}
\label{figtheorem2}

\vspace{-1.5em}
\end{figure}

\subsubsection{Rechecking}Finally, we apply the error correction method used in Stage \uppercase\expandafter{\romannumeral 2} to the expanded pseudo seed set and make further corrections to obtain the final pseudo seed set $S3$.

\section{Theoretical Analysis}
\label{analysis}
In this section, we focus on analyzing the effect of the quality of pseudo seeds on a multimodal entity alignment model (such as MCLEA, MEAformer) based on contrastive learning.

To illustrate the ICL optimization process, we conducted an in-depth analysis of the lower bound of ICL, highlighting key factors influencing the optimization process.

\begin{theorem}
\label{theorem1}
     Assuming that a regularization function is applied to the embedding representation of an entity, let $S = \{ (h_1^1,h_2^1), \ldots, (h_1^n,h_2^n)\}$ be the set of pseudo seeds where $h_1^n$ denotes the vector representation of entity $n$ in knowledge graph $1$ and $ ||h|| = 1$. Then, the lower bound of ICL of batch $B$ is:
 \begin{equation}
{L^{ICL}} \ge \sum\limits_{i \in B} {\log (1 + } D\exp (\underbrace {\frac{1}{2}||h_1^i - h_2^i||}_{attraction} - \underbrace {\frac{1}{{2D}}\sum\limits_{j \in {\mathcal N}_i^{ng}} {||h_1^i - {h^j}||} }_{repulsion} - \frac{1}{D}))
 \end{equation}
 where ${\mathcal N}_i^{ng}$ denotes the set of negative samples for entity i, and $D=|B|-1$ corresponds to the number of negative instances in the mini-batch, excluding the anchor entity itself.
\end{theorem}

\textit{Proof: See Appendix~\ref{proof1}.}

\begin{table*}[]
\centering
\caption{Comparison of applying PSQE to four recent representative open-source unsupervised MMEA methods (EVA, MCLEA, MEAformer, PCMEA) on the DBP15K dataset.
`MEAformer+PSQE' denotes the application of PSQE to MEAformer.
`Unsup' indicates whether the method is unsupervised or not.
It should be noted that the experimental results of EVA, MCLEA, MEAformer, and PCMEA deviate from the original paper in terms of specific values due to the differences in the setting of the initialization seeds.
}

\begin{tabular}{cccccccccccc}
\hline
\multirow{2}{*}{}       & \multirow{2}{*}{Model} & \multirow{2}{*}{Unsup} & \multicolumn{3}{c}{\textit{DBP15K}$_{ZH-EN}$}           & \multicolumn{3}{c}{\textit{DBP15K}$_{JA-EN}$}           & \multicolumn{3}{c}{\textit{DBP15K}$_{FR-EN}$}           \\
                        &                        &                        & H@1           & H@10          & MRR           & H@1           & H@10          & MRR           & H@1           & H@10          & MRR           \\ \hline

\multirow{8}{*}   & RDGCN {\tiny (IJCAI, 2019)}                 & ×                      & .708          & .846          & .746          & .767          & .895          & .812          & .886          & .957          & .911          \\
                        & RPR-RHGT {\tiny (IJCAI, 2022)}             & ×                      & .693          & -             & .754          & .886          & -             & .912          & .889          & -             & .919          \\
                        & DESAlign {\tiny (ICDE, 2024)}             & x                      
                        & .826          & \textbf{.972}          & \textbf{.885}         & .811          & .963      & .869          & .810          & .957          & .865          \\
                        \hdashline
                        
                        & EVA {\tiny (AAAI, 2021)}                   & \checkmark
                        & .783          & .898          & .825          & .861          & .946          & .892          & .914          & .972          & .936          \\
                        \rowcolor{gray!20}
                         & \textbf{EVA+PSQE}                   & \checkmark          
                        & .823 & .930 
                        & .861 & .873 
                        & .955 & .904 
                        & .928 & .977 
                        & .948 \\
                        \hdashline
                        & MCLEA {\tiny (COLING, 2022)}                   & \checkmark                      & .792          & .906          & .834          & .866          & .956          & .900          & .908          & .972          & .933          \\
                        \rowcolor{gray!20}
                        & \textbf{MCLEA+PSQE}                 & \checkmark                      %& .835          & .940          & .874          & .879          & .964          & .911          & .929          & \textbf{.987} & .951          \\
                        & .835 & .940 
                        & .874 & .879 
                        & .964 & .911 
                        & .929 &  \textbf{.987} 
                        & .951 \\
                        \hdashline
                        & MEAformer {\tiny (MM, 2023)}             & \checkmark                      
                        & .804          & .923          & .847          & .872          & .960      & .905          & .918          & .980          & .942          \\
                        \rowcolor{gray!20}
                        & \textbf{MEAformer+PSQE }            & \checkmark                     % & \textbf{.842} & \textbf{.947} & \textbf{.882} & \textbf{.892} & \textbf{.968} & \textbf{.921} & \textbf{.932} & .986          & \textbf{.953} \\ \hline
                        & \textbf{.842} 
                        & .947 
                        & .882 
                        & \textbf{.892} 
                        & \textbf{.968} 
                        & \textbf{.921} 
                        & \textbf{.932} 
                        & .986 
                        & \textbf{.953} \\ 
                        \hdashline
                        & PCMEA {\tiny (AAAI, 2024)}             & \checkmark                      
                        & .759          & .894          & .807          & .833          & .942      & .873          & .904          & .974          & .930          \\
                        \rowcolor{gray!20}
                        & \textbf{PCMEA+PSQE }            & \checkmark                     
                        & .816
                        & .934 
                        & .859 
                        & .860
                        & .956
                        & .895
                        & .917 
                        & .981
                        & .941 \\ 
\hline
\end{tabular}%

\label{dbpQ1}

\end{table*}

The derivation of the above theory was inspired by~\cite{khosla2020supervised, zhu2022balanced}, theorem~\ref{theorem1} shows that the lower bound of ICL consists of two components: an attraction term and a repulsion term.

\begin{remark}
The attraction term enforces the representations of positive samples in pseudo-alignment seeds to be as similar as possible. Therefore, the precision of the pseudo-alignment seeds plays a critical role in the effectiveness of model training.
\label{remark1}
\end{remark}

To reduce the overall loss, the attraction term should be minimized. That is, the representations of the positive samples in the pseudo-aligned seeds should be as similar as possible. However, when the pseudo-aligned seed pairs are mismatched, the attraction term drives the two entities that should not have been aligned to be constantly close to each other, thus increasing the difficulty for the model in distinguishing the correct entity pairs.
For example, Fig.~\ref{figtheroem1} (b) illustrates the case of incorrect entity pairs versus correct entity pairs. When incorrect entity pairs occur, they lead to incorrect attraction, which essentially pushes the vector representations of correct entity pairs apart, causing the model to introduce bias.
Therefore, to improve the precision of the model, it is necessary to improve the precision of the pseudo-aligned seed pairs and reduce the occurrence of false matches.  
For example, Fig.~\ref{figtheroem1} (a) shows that when a correct entity pair appears, it actually corresponds to distancing from the negative sample.
\begin{remark}
    The repulsion term adjusts the distance between entities in pseudo seeds and their negative samples. 
    When the coverage of the negative sample graph is unbalanced, a more concentrated set of entities distributed in the graph produces a more significant gradient, which makes the loss function more inclined to focus on these entities and can lead to a biased representation of the feature space in the model during training. Therefore, pseudo-seeds need to be evenly covered in the graph to avoid the under-optimization of sparse entities.
\vspace{-1em}
\end{remark}
\textit{Proof: The gradient update process is detailed in Appendix~\ref{proof2}.}

To reduce the overall loss, the effect of the repulsion term needs to be as large as possible. That is, the differences between negative samples in pseudo-alignment seeds should be maximized to effectively distinguish the representations of different entities. However, issues arise when the graph coverage of entities in pseudo-alignment seeds is uneven. For example, as shown in Fig.~\ref{figtheorem2} (a), some negative sample entities are scattered, while in Fig.~\ref{figtheorem2} (b), other entities are more concentrated (with negative samples including $x_9^-$, $x_7^-$). This results in more concentrated entities (such as $x$, $x_1^-$, $x_2^-$, $x_3^-$, $x_4^-$) having a larger weight in the repulsion term. This weight difference can lead to the gradient of concentrated entities being significantly larger than that of scattered entities, causing the loss function to focus more on the concentrated entities during optimization. Therefore, entities (such as $x_7^-$, $x_8^-$, $x_9^-$) can become harder to distinguish in the embedding space, leading to lower model training performance.

Analyzing from the theoretical level, PSQE strengthens the optimization effect of the attraction term and the repulsion term in Theorem~\ref{theorem1}.
PSQE enhances the precision of pseudo-seeds, thereby achieving a tighter alignment of entities and optimizing the attraction term. Meanwhile, by balancing the distribution of negative samples, it optimizes the repulsion term.

\begin{table}[htbp]
\centering
\caption{Comparison results on the DWY15K dataset.}
\resizebox{0.47\textwidth}{!}{%
\begin{tabular}{ccccccc}
\hline
\multirow{2}{*}{Method} & \multicolumn{3}{c}{DW-V1}                & \multicolumn{3}{c}{DW-V2}                \\  
                        & H@1           & H@10          & MRR           & H@1           & H@10          & MRR           \\ \hline

COTSAE \tiny{(AAAI, 2020)}                 &.709          & .904          & .77           & .922          & .983          & .940          \\
\hdashline
EVA                     & .880          & .965          & .914          & .880          & .963          & .913          \\
\rowcolor{gray!20}
\textbf{EVA+PSQE}           %& .915 & .971 & .937 & .900 & .972 & \.929 \\
                       &.915 & .971 & .937  &.900 & .972 & .929 \\
\hdashline
MCLEA                   & .903          & .975          & .932          & .896          & .970          & .926          \\
\rowcolor{gray!20}
\textbf{MCLEA+PSQE}         %& .912 & .975 & .937 & .921 & .977 & \.943 \\
                        & .912 
                        & .975 
                        & .937 
                        & .921 
                        & .977 
                        & .943 \\
  \hdashline          
MEAformer               & .928          & .983          & .950          & .917          & .981          & .944          \\ \hline

\rowcolor{gray!20}
\textbf{MEAformer+PSQE}     %& \textbf{.954} & \textbf{.984} & \textbf{.966} & \textbf{.939} & \textbf{.984} & \textbf{.957} \\ \hline
                        & \textbf{.954}
                        & \textbf{.984} 
                        & \textbf{.966} 
                        & \textbf{.939} 
                        & \textbf{.984} 
                        & \textbf{.957} \\
                        \hline

\end{tabular}%
}
\label{DWv1Q1}
\vspace{-0.5em}
\end{table}

\section{Experiments}
\subsection{Experimental settings}
\subsubsection{Datasets}
Our experiments are conducted on two types of real-world multimodal datasets: DBP15K~\cite{sun2017cross} and DWY15K~\cite{guo2019learning}.
DBP15K is a widely used cross-language multimodal entity alignment benchmark dataset containing three bilingual entity alignment benchmarks: French-English (FR-EN), Japanese-English (JA-EN), and Chinese-English (ZH-EN).
DWY15k is a monolingual multimodal entity alignment dataset, which is an EA focusing on DBpedia-Wikidata, and mainly contains two subsets with different levels of sparsity: DW-V1 and DW-V2. DW-V1 has a higher sparsity than DW-V2.
% See Appendix Table~\ref{dataset} for detailed data set statistics.
Given that the existing multimodal unsupervised methods have not yet been experimented on the FB15K-DB15K~\cite{MMKG} and FB15K-YAGO15K datasets~\cite{MMKG}, the relevant experimental validations are not included in this paper.

\subsubsection{Evaluation Metrics}
The most commonly used metrics for evaluating the performance of knowledge graph embedding models are $Hits@\textit{n}$ and $MRR$. Higher values of $Hits@n$ ($H@n$) and $MRR$ indicate better performance. 
Hits@\textit{n} denotes the average percentage of triples ranked less than or equal to \textit{n} in link prediction and is calculated as follows:
\begin{equation}
{{Hits}}@n = \frac{1}{{|S|}} \sum\limits_{i = 1}^{|S|} \left( rank_i \le n \right),
\end{equation}
where \( S \) is the set of triples, \( |S| \) denotes the number of elements in this set of triples, and \( rank_i \) is the link prediction rank of the \( i \)th triple. MRR represents the mean reciprocal rank of correct entities, calculated by the following formula:
\begin{equation}
\scalebox{0.9}{$
\begin{aligned}
{MRR} &= \frac{1}{|S|} \sum\limits_{i = 1}^{|S|} \left(\frac{1}{{rank}_i}\right) \\
           &= \frac{1}{|S|} \left( \frac{1}{{rank}_1} + \ldots + \frac{1}{{rank_{|S|}}} \right)
\end{aligned}
$}
\label{pd}
\end{equation}

\textbf{Seed Quality Evaluation.}
In order to better measure the quality of pseudo seeds, we evaluate them in terms of both their precision and graph coverage  balance, which can be formalized as:
\begin{equation}
    \left\{ {\begin{array}{*{20}{c}}
{precise = \frac{{{S_t}}}{S}}\\
{graph coverage  = \frac{{{S_a}}}{{2*{G_n}}} + \frac{{Edge}}{{2*G\_Edge}} + \frac{{{S_f}}}{{{G_n}}}}
\end{array}} \right.
\label{metric}
\end{equation}
Where $S_t$ denotes the number of correct seeds in the pseudo seeds, $S$ denotes the total number of pseudo seeds, $S_a$ denotes the number of aggregated entities in the pseudo seed, $S_f$ denotes the number of scattered entities in the pseudo seed, and $Edge$ and $G\_Edge$ denote the number of edges that the entities in the pseudo seed contain in the graph, as well as the total number of all edges in the graph, respectively.
When all of the pseudo seeds are correct, the precision value reaches 1. When the number of pseudo seeds reaches the total number of entities in the atlas, the graph coverage value will reach 1.

\subsubsection{Implementation Details}
All experiments were conducted on NVIDIA A100 GPUs.
To ensure the reproducibility of the experiments, we set the random number seed to 42, the visual coder is ResNet-152 which is consistent with EVA, MCLEA, and MEAformer, and the visual feature dimension is 2048, the number of initialized seeds is set to 1000 on the three datasets of DBP15K, and the number of initialized seeds is set to 5000 on the two datasets of DWY15K. In Stage \uppercase\expandafter{\romannumeral 1}, the search space for the number of clusters is from 2 to 5, and the weights of each modality in generating the seeds are 0.8 for vision, 0.1 for a relationship, and 0.1 for the attribute. In Stage \uppercase\expandafter{\romannumeral 2}, the model's training rounds are 300, the modal vector dimensions are 300, the learning rate is 0.01, and the batch size is 2000. In Stage \uppercase\expandafter{\romannumeral 3}, the neighborhood seed similarity threshold is initially set to 0.8. For the entities that do not have a visual, we agree with the methods~\cite{liu2021visual,mclea} by assigning them a random vector sampled from a normal graph coverage parameterized by the mean and standard deviation of the other visuals.
\subsection{Baselines}
The baseline methods include nine recent entity alignment algorithms: RDGCN~\cite{ijcai2019p733}, RPR-RHGT~\cite{cai2022entity}, COTSAE~\cite{yang2020cotsae}, DESAlign~\cite{desalign}, EVA~\cite{liu2021visual}, MCLEA~\cite{mclea}, MEAformer~\cite{chen2023meaformer}, and PCMEA~\cite{wang2024pseudo}.
Among them, RDGCN, COTSAE, and RPR-RHGT belong to unimodal supervised entity alignment methods, while DESAlign, EVA, MCLEA, and MEAformer are multimodal entity alignment methods. Since there are fewer studies on unsupervised multimodal entity alignment and mainstream methods mainly focus on EVA, MCLEA, PCMEA, and MEAformer, our experiments are to investigate the improvement of our proposed PSQE with these four methods. Specifically, we use PSQE to generate seeds and use these four methods to obtain the final results.
In addition, to ensure the consistency of the experimental configurations, we experimentally compare MEAformer, MCLEA, PCMEA, and EVA under a uniform setup, while the experimental results of the remaining compared methods are directly quoted from the original paper.

\vspace{-0.5em}

\subsection{Overall Results}

\begin{tcolorbox} [
    colback=blue!5!white,
    boxsep=0.3pt 
] %个最朴素的 tcolorbox 环境
RQ1: How well does PSQE boost the pseudo seed quality on multimodal entity alignment models?
\end{tcolorbox}

To address Research Question 1, we conducted experiments on five sub-datasets (JA-EN, ZH-EN, FR-EN, DW-V1, DW-V2) and analyzed the effectiveness of our algorithms through specific cases. The experimental results show that PSQE significantly improves the quality of pseudo seeds and indirectly improves the effectiveness of existing unsupervised entity alignment methods.
\subsubsection{Main results}
As shown in Table~\ref{dbpQ1} and Table~\ref{DWv1Q1}, we conducted experiments on five multimodal sub-datasets using three classic unsupervised multimodal entity alignment models (i.e., MEAformer, EVA, MCLEA) with the pseudo seeds generated by the PSQE strategy. The results indicate that applying the PSQE strategy improves the training performance of all three multimodal entity alignment models. For instance, in Table~\ref{dbpQ1}, When surface information (e.g., entity names and character information) is utilized in the MEAformer, the application of PSQE on the MEAformer model improved the $H@1$ score on the ZH-EN, JA-EN, and FR-EN datasets from 80.4\%, 87.2\%, and 91.8\% to 84.2\%, 89.2\%, and 93.2\%, respectively, corresponding to improvements of 3.8\%, 2.0\%, and 1.4\%. In Table~\ref{DWv1Q1}, for the two sub-datasets of DWY15K, the application of PSQE to the EVA, MCLEA, and MEAformer models led to an increase of over 0.8\% in terms of the $H@1$ score. These results suggest that PSQE, by enhancing the quality of pseudo seeds, can effectively promote the models to learn more accurate and richer information across different datasets and MMEA models. 
\begin{figure}[htbp]
\centering
\includegraphics[width=\linewidth]{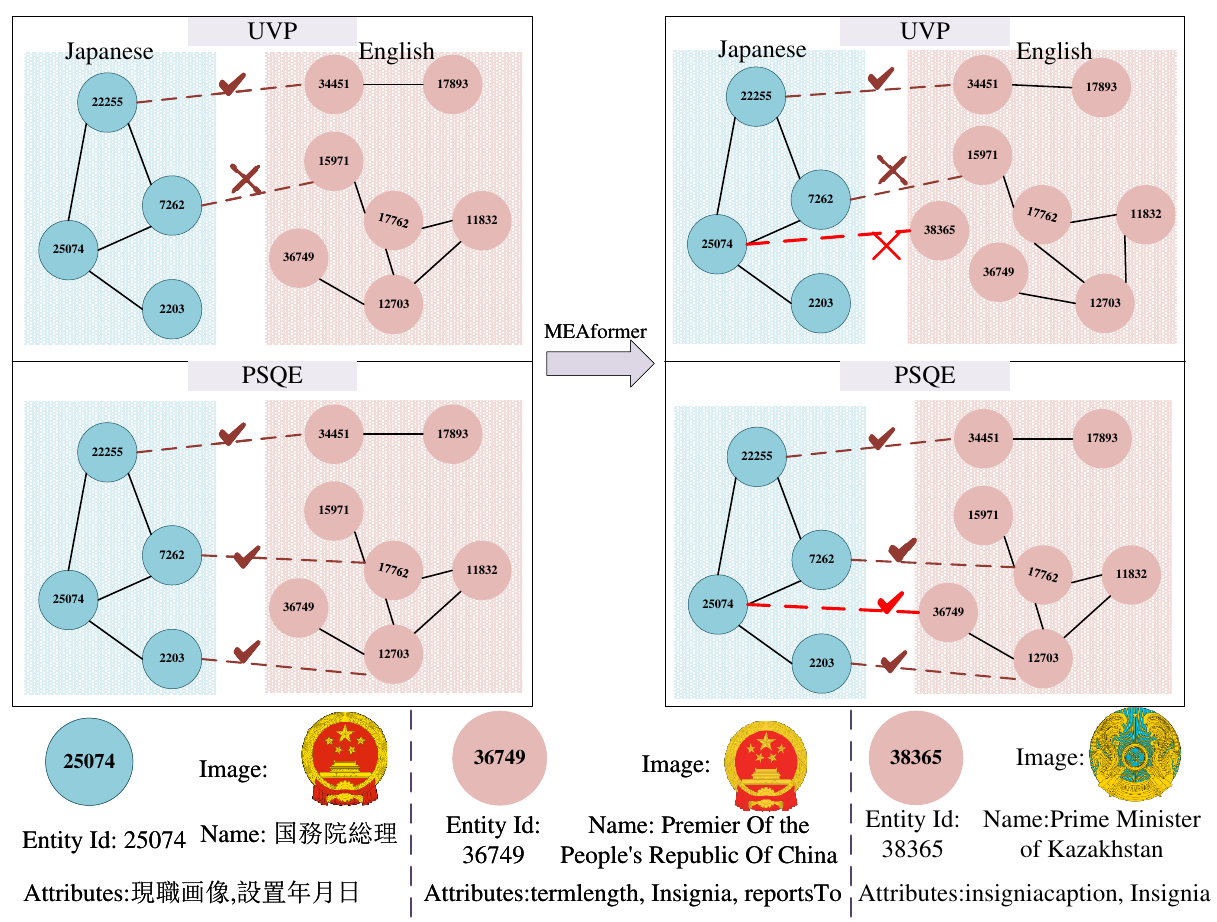}
\caption{Case study of the results of MEAformer's comparison on the unsupervised setting of PSQE vs. UVP on the JA-EN dataset, where correct entity pairs are marked with `\checkmark' and incorrect ones with `×'.}
\label{case_study}
\vspace{-1em}
\end{figure}

\subsubsection{Case study}
Fig.~\ref{case_study} shows a specific case of the comparison between PSQE and UVP to generate quality seeds on the JA-EN dataset using the MEAformer model. In the local knowledge graph in the figure, the pseudo seed pairs generated by UVP include `22255-34451' and `7262-15971', while the pseudo seed pairs generated by PSQE include `22255-34451', `7262-17762', and `2203-12703'. In subsequent model training, UVP produced incorrect matches. For instance, `25074' was incorrectly associated with `38365: Prime Minister of Kazakhstan', whereas PSQE, on the other hand, caused MEAformer to correctly recognize `36749: Premier of the People's Republic of China'. This is mainly because UVP generates the wrong pseudo seed pair `7262-15971' with a limited graph coverage of pseudo seed pairs, which fails to effectively guide the structural information of the figure in the model.
PSQE, on the other hand, ensures the correct transmission of graph structural information through accurate pseudo seed pairs (e.g., `2203-12703'), thus avoiding wrong predictions during training and improving the model performance.

\begin{figure}[]
\centering
\includegraphics[width=\linewidth]{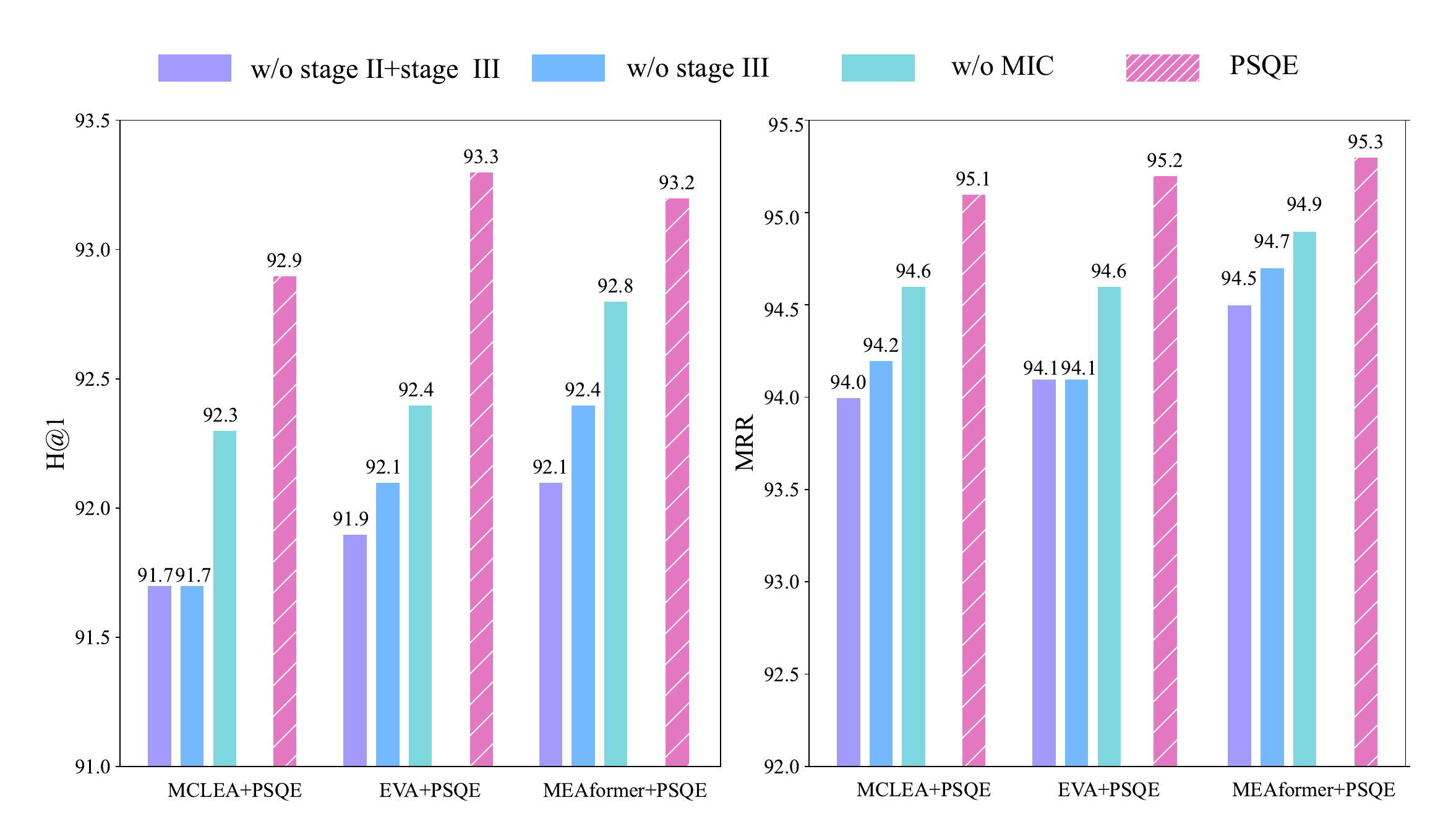}
\caption{Ablation study of PSQE on the FR-EN dataset. }
\label{RQ22}
\vspace{-1em}
\end{figure}

\subsection{Ablation Study}
\begin{tcolorbox} [
    colback=blue!5!white,
    boxsep=0.3pt 
] %个最朴素的 tcolorbox 环境
RQ2: How do the precision and the graph coverage of pseudo seeds affect the alignment results of MMEA? 
Which modality has the most significant impact on the pseudo seed generation process?
\end{tcolorbox}
% To address Research Question 2, we analyzed the impact of different components of PSQE on the results. Fig.~\ref{RQ22} and Table~\ref{tab:comparison} show the results of the ablation experiments on the FR-EN and JA-EN, respectively. The experimental results show that both the accuracy and graph coverage of pseudo seeds are important and that the visual information helps to improve the model performance.

\begin{table}[!t]
\centering
\caption{\textbf{Component analysis} for PSQE on MEAformer, EVA, and MCLEA models under the JA-EN dataset, where ``w/o rel” denotes the removal of relational information from the PSQE, ``w/o attr” denotes the removal of the attribute information and ``w/o visual'' denotes the removal of the visual information.}

\begin{tabular}{cccccll}
\hline
\multirow{2}{*}{} & \multicolumn{2}{c}{MEAformer} & \multicolumn{2}{c}{EVA} & \multicolumn{2}{c}{MCLEA} \\
\cmidrule(lr){2-3} \cmidrule(lr){4-5} \cmidrule(lr){6-7}
 & H@1 & MRR & H@1 & MRR & \multicolumn{1}{c}{H@1} & \multicolumn{1}{c}{MRR} \\ \hline
PSQE & \textbf{.892} & \textbf{.921} & \textbf{.878} & \textbf{.907} & \textbf{.879} & \textbf{.911} \\
PSQE (w/o rel) & .889 & .919 & .870 & .901 & .877 & .908 \\
PSQE (w/o attr) & .891 & .921 & .871 & .902 & .877 & .909 \\ 
PSQE (w/o visual) & .732 & .795 &.603 &.667 & .668 & .735 \\
\hline

\end{tabular}
\label{tab:comparison}
\vspace{-1em}
\end{table}

\subsubsection{The importance of pseudo seed precision.} As shown in Fig.~\ref{RQ22}, in the scenario of exploring pseudo seeded precision, we designed several experiments, i.e., removing some multimodal information, which includes removing the multimodal error correction part of Stages \uppercase\expandafter{\romannumeral 2} and \uppercase\expandafter{\romannumeral 3} (\textbf{w/o MIC}). Our results show that both removing multimodal error correction and using only unimodal information negatively affect the final alignment results. For example, with the MEAformer, the removal of multimodal error correction resulted in a decrease of $H@1$ by 0.4\%, respectively. These results indicate that the multimodal information has a good effect on improving the precision of pseudo seeds and that the precision of pseudo seeds is a part of the quality of pseudo seeds.

\subsubsection{Importance of pseudo seed graph coverage balance.} In the scenarios where we studied the pseudo seed graph coverage, we mainly designed the stage of removing PSQE to control the impact of the graph coverage balance, which includes experiments with removing Stage \uppercase\expandafter{\romannumeral 3} and removing Stage \uppercase\expandafter{\romannumeral 2}+Stage \uppercase\expandafter{\romannumeral 3}. In the case of removing Stage \uppercase\expandafter{\romannumeral 3}, the $MRR$ decreased by \textbf{1.1\%, 0.9\% } and \textbf{0.6\%} when PSQE was applied to EVA, MCLEA, and MEAformer, respectively.
And with Stage \uppercase\expandafter{\romannumeral 2} + Stage \uppercase\expandafter{\romannumeral 3} removed, $MRR$ decreased by \textbf{1.1\%}, \textbf{1.1\%} and \textbf{0.8\%}, respectively. The experimental results show that both Stage \uppercase\expandafter{\romannumeral 2} and Stage \uppercase\expandafter{\romannumeral 3} can be effective by enhancing the graph coverage and that the coverage distribution of pseudo seeds is equally an important component of the pseudo seed quality.
Although Stage II in MCLEA has little effect on $H@1$, it significantly improves efficiency and pseudo seed quality, reduces seed dependence, and maintains or enhances overall performance.
%, with details provided in the Appendix~\ref{Impact of Stage II}.

\subsubsection{Importance of visual information.}
Table ~\ref{tab:comparison} shows the performance comparison of PSQE with MEAformer, EVA, and MCLEA on the JA-EN dataset after removing attribute modality, relational modality, and visual modality, respectively. The experimental results show that removing attribute modality or relational modality has a relatively small impact on the overall performance. For example, the $H@1$ index of MEAformer decreases by only 0.3\% after removing the relational mode, while the $H@1$ index of EVA decreases by 0.7\% after removing the attribute mode. In contrast, when the visual modality is removed, the performance of the methods decreases significantly, e.g., MEAformer's $H@1$ decreases by 16\%, while MCLEA's decreases by 11.1\%.
This phenomenon suggests that although attribute and relational modalities provide useful complements in multimodal information fusion, visual modalities play a more critical role in enhancing the quality of entity representations. The reason for this discrepancy may lie in the fact that the number of categories of attribute and relational information is more limited, making it difficult to provide sufficient discriminative properties. For example, in JA, there are only about 1,299 relationship types compared to 15,000 entities, limiting the model’s ability to distinguish entities based on semantic information alone. In contrast, the 12,739 entity images provide rich, high-dimensional visual features that enhance the discriminative power of entity embeddings~\cite{wang2023probing}.

\subsection{Parameter Analysis}
\begin{tcolorbox}[
    colback=blue!5!white,
    boxsep=-0.1pt 
]
\textbf{RQ3:} How does the stability of PSQE perform under different parameter scenarios?
\end{tcolorbox}

% To address research question 3, we designed two experiments in unsupervised scenarios to analyze the effects of \textbf{different numbers of clusters}. 
% Our experimental results show that the number of clusters has a small impact on the final model effect. While the alignment effect improves as the number of initialization seeds increases, the overall impact is limited. 
To investigate the impact of the number of clusters in stage \uppercase\expandafter{\romannumeral 1}, we evaluated the performance of PSQE with the three baseline algorithms under different clustering settings on the FR-EN dataset, as shown in Fig.~\ref{CLUSter}. The results indicate that variations in the number of clusters have a slight impact on the precision of different models with minimal fluctuations. For example, when PSQE is applied to EVA, the $H@1$ values across different clustering settings range from 92.7\% to 93.2\% with a mere \textbf{0.5\%} difference and no significant fluctuations. This suggests that the performance of PSQE is not particularly sensitive to changes in the number of clusters. 

\begin{figure}[!t]
\centering
\includegraphics[width=\linewidth]{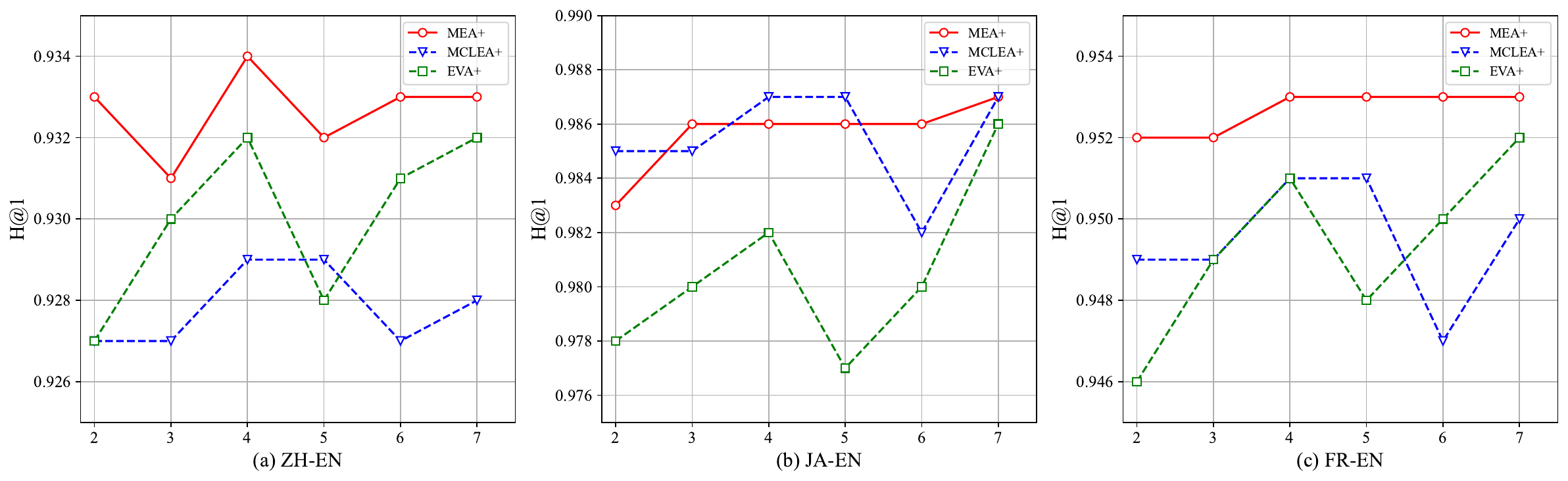}
\caption{The impact of different clustering numbers in Stage \uppercase\expandafter{\romannumeral 1} on the PSQE performance on the FR-EN dataset.}
\label{CLUSter}
\end{figure}

\section{Conclusion}
In this paper, we propose PSQE (Pseudo-Seed Quality Enhancement), a novel theoretical-practical framework for unsupervised MMEA that jointly optimizes seed precision and graph coverage balance. PSQE leverages multimodal information (text, images, and attributes) alongside a clustering-resampling strategy to generate high-quality pseudo seeds in three stages: (1) multimodal embedding and clustering to ensure balanced coverage, (2) feature resampling with error correction to enhance precision, and (3) neighborhood sampling to refine alignment. Theoretically, we analyze how pseudo seeds influence contrastive learning models, demonstrating that seed precision governs the attraction term (aligning correct pairs), while coverage balance regulates the repulsion term (separating negative samples). Imbalanced coverage skews optimization toward dense graph regions, degrading performance for sparse entities—a phenomenon empirically validated by our experiments. We validate the effectiveness of our proposed method by applying it to three representative unsupervised approaches across two large-scale public benchmarks. The experimental results on public datasets demonstrated the effectiveness in multimodal entity alignment tasks.%In this paper, we provided the first systematic analysis of the factors affecting the quality of pseudo-alignment seeds when using contrastive learning methods in unsupervised multimodal entity alignment. Our study shows that the graph coverage  balance and the precision of pseudo-alignment seeds have important impacts on model training and that low-quality pseudo seeds can lead to a degradation in the quality of the representation of entity embeddings. To address this, we proposed the PSQE strategy, which enhances seed precision with multimodal information and optimizes seed graph coverage  balance through a three-stage generation process. The experimental results on public datasets demonstrated the effectiveness of PSQE in multimodal entity alignment tasks.

\section{Acknowledgements}
This work was supported by the National Natural Science Foundation of China (Grant No. 62120106008), Anhui Provincial Science and Technology Fortification Plan (Grant No. 202423k09020015), the Hefei Key Generic Technology Research and Development Program (No. 2024SGJ010), the Youth Talent Support Program of the Anhui Association for Science and Technology (Grant No. RCTJ202420), New Chongging Youth Innovation Talent Project under Grant  CSTB2024NSCQ-QCXMX0035, and the Open Project of Key Laboratory of Knowledge Engineering with Big Data (the Ministry of Education of China), under grant number BigKEOpen2025-03. The computation was completed on the HPC Platform of Hefei University of Technology. Y. He was not supported by any of these funds.
\clearpage

\bibliographystyle{ACM-Reference-Format}
\bibliography{sample-base}

@inproceedings{bu2025query,
  title={Query-Driven Multimodal GraphRAG: Dynamic Local Knowledge Graph Construction for Online Reasoning},
  author={Bu, Chenyang and Chang, Guojie and Chen, Zihao and Dang, CunYuan and Wu, Zhize and He, Yi and Wu, Xindong},
  booktitle={Findings of the Association for Computational Linguistics: ACL 2025},
  pages={21360--21380},
  year={2025}
}

@inproceedings{huang2025mitigate,
  title={How to Mitigate Information Loss in Knowledge Graphs for GraphRAG: Leveraging Triple Context Restoration and Query-Driven Feedback},
  author={Huang, Manzong and Bu, Chenyang and He, Yi and Wu, Xindong},
 booktitle    = {Proceedings of the Thirty-Fourth International Joint Conference on Artificial Intelligence},
  pages        = {8104--8112},
publisher    = {ijcai.org},
  year         = {2025}
}

@InProceedings{MMKG,
author="Liu, Ye
and Li, Hui
and Garcia-Duran, Alberto
and Niepert, Mathias
and Onoro-Rubio, Daniel
and Rosenblum, David S.",
title="Mmkg: Multi-modal knowledge graphs",
booktitle="Proceedings of The Semantic Web",
year="2019",
publisher="Springer",
pages="459--474"
}

@article{liu2024cross,
  title={Cross-model cross-stream learning for self-supervised human action recognition},
  author={Liu, Mengyuan and Liu, Hong and Guo, Tianyu},
  journal={IEEE Transactions on Human-Machine Systems},
  year={2024},
  publisher={IEEE}
}

@INPROCEEDINGS{desalign,
  author={Wang, Yuanyi and Sun, Haifeng and Wang, Jiabo and Wang, Jingyu and Tang, Wei and Qi, Qi and Sun, Shaoling and Liao, Jianxin},
  booktitle={ Proceedings of International Conference on Data Engineering }, 
  title={Towards semantic consistency: Dirichlet energy driven robust multi-modal entity alignment}, 
  year={2024},
  volume={},
  number={},
  pages={3559-3572},
}

@article{hamerly2003learning,
  title={Learning the k in k-means},
  author={Hamerly, Greg and Elkan, Charles},
  journal={Advances in neural information processing systems},
  volume={16},
  year={2003}
}

@inproceedings{tanwar-etal-2023-multilingual,
    title = "Multilingual llms are better cross-lingual in-context learners with alignment",
    author = "Tanwar, Eshaan  and
      Dutta, Subhabrata  and
      Borthakur, Manish  and
      Chakraborty, Tanmoy",
    booktitle = "Proceedings of the Association for Computational Linguistics",
    year = "2023",
    publisher = "ACL",
    pages = "6292--6307",
}

@inproceedings{zhang2025robust,
  title={An Robust Entity Alignment Method based on Knowledge Distillation with Noisy Aligned Pairs},
  author={Zhang, Yuhong and Song, Hangchi and Zhu, Xiaolong and Bu, Chenyang and Yu, Kui},
  booktitle={Proceedings of ACM International Conference on Information and Knowledge Management},
  pages={5510--5514},
publisher = {ACM},
  year={2025}
}

@ARTICLE{10948482,
  author={Li, Jiatong and Liu, Wei and Ding, Zhihao and Fan, Wenqi and Li, Yuqiang and Li, Qing},
  journal={IEEE Transactions on Knowledge and Data Engineering (Early Access)}, 
  title={Large language models are in-context molecule learners}, 
  year={2025},
  volume={},
  number={},
  pages={1-13},
  doi={10.1109/TKDE.2025.3557697}
}

@article{khosla2020supervised,
  title={Supervised contrastive learning},
  author={Khosla, Prannay and Teterwak, Piotr and Wang, Chen and Sarna, Aaron and Tian, Yonglong and Isola, Phillip and Maschinot, Aaron and Liu, Ce and Krishnan, Dilip},
  journal={Advances in neural information processing systems},
  volume={33},
  pages={18661--18673},
  year={2020}
}

@inproceedings{zhu2017iterative,
  title={Iterative entity alignment via joint knowledge embeddings.},
  author={Zhu, Hao and Xie, Ruobing and Liu, Zhiyuan and Sun, Maosong},
 booktitle ={Proceedings of International Joint Conference on Artificial Intelligence},
  publisher={IJCAI Oragnization},
  volume={17},
  pages={4258--4264},
  year={2017}
}

@article{mak,
  title={Improving contrastive learning on imbalanced data via open-world sampling},
  author={Jiang, Ziyu and Chen, Tianlong and Chen, Ting and Wang, Zhangyang},
  journal={Advances in Neural Information Processing Systems},
  volume={34},
  pages={5997--6009},
  year={2021}
}

@inproceedings{ack-mmea,
author = {Li, Qian and Guo, Shu and Luo, Yangyifei and Ji, Cheng and Wang, Lihong and Sheng, Jiawei and Li, Jianxin},
title = {Attribute-consistent knowledge graph representation learning for multi-modal entity alignment},
year = {2023},
publisher = {ACM},
booktitle = {Proceedings of the ACM Web Conference},
pages = {2499–2508},
numpages = {10}
}

@inproceedings{liu2022selfkg,
  title={Selfkg: Self-supervised entity alignment in knowledge graphs},
  author={Liu, Xiao and Hong, Haoyun and Wang, Xinghao and Chen, Zeyi and Kharlamov, Evgeny and Dong, Yuxiao and Tang, Jie},
  booktitle={Proceedings of the ACM web conference},
  pages={860--870},
publisher = {ACM},
  year={2022}
}

@article{zhu2023mmiea,
  title={Mmiea: Multi-modal interaction entity alignment model for knowledge graphs},
  author={Zhu, Bin and Wu, Meng and Hong, Yunpeng and Chen, Yi and Xie, Bo and Liu, Fei and Bu, Chenyang and Ding, Weiping},
  journal={Information Fusion},
  volume={100},
  pages={101935},
  year={2023},
  publisher={Elsevier}
}

@article{wang2023probing,
  title={Probing the impacts of visual context in multimodal entity alignment},
  author={Wang, Meng and Shi, Yinghui and Yang, Han and Zhang, Ziheng and Lin, Zhenxi and Zheng, Yefeng},
  journal={Data Science and Engineering},
  volume={8},
  number={2},
  pages={124--134},
  year={2023},
  publisher={Springer}
}

@inproceedings{MSNEA,
  title={Multi-modal siamese network for entity alignment},
  author={Chen, Liyi and Li, Zhi and Xu, Tong and Wu, Han and Wang, Zhefeng and Yuan, Nicholas Jing and Chen, Enhong},
  booktitle={Proceedings of SIGKDD Conference on Knowledge Discovery and Data Mining},
  pages={118--126},
  year={2022},
  publisher = {ACM},
}

@inproceedings{chen2023meaformer,
  title={Meaformer: Multi-modal entity alignment transformer for meta modality hybrid},
  author={Chen, Zhuo and Chen, Jiaoyan and Zhang, Wen and Guo, Lingbing and Fang, Yin and Huang, Yufeng and Zhang, Yichi and Geng, Yuxia and Pan, Jeff Z. and Song, Wenting and Chen, Huajun},
  booktitle={Proceedings of International Conference on Multimedia},
  pages={3317--3327},
publisher = {ACM},
  year={2023}
}

@inproceedings{cai2022entity,
  title={Entity alignment with reliable path reasoning and relation-aware heterogeneous graph transformer},
  author={Cai, Weishan and Ma, Wenjun and Zhan, Jieyu and Jiang, Yuncheng},
  booktitle = {Proceedings of International Joint Conference on Artificial Intelligence},
  publisher = {IJCAI Organization},
  pages={1930–193},
  year={2022},

}

@inproceedings{sun2017cross,
  title={Cross-lingual entity alignment via joint attribute-preserving embedding},
  author={Sun, Zequn and Hu, Wei and Li, Chengkai},
  booktitle={Proceedings of The Semantic Web--ISWC},
  pages={628--644},
  year={2017},
  organization={Springer}
}

@inproceedings{liu2021visual,
  title={Visual pivoting for (unsupervised) entity alignment},
  author={Liu, Fangyu and Chen, Muhao and Roth, Dan and Collier, Nigel},
  booktitle={Proceedings of the AAAI Conference on Artificial Intelligence},
  volume={35},
  number={5},
  pages={4257--4266},
publisher = {AAAI Press},
  year={2021}
}

@inproceedings{ijcai2019p733,
  title     = {Relation-aware entity alignment for heterogeneous knowledge graphs},
  author    = {Wu, Yuting and Liu, Xiao and Feng, Yansong and Wang, Zheng and Yan, Rui and Zhao, Dongyan},
  booktitle = {Proceedings of International Joint Conference on Artificial Intelligence},
  publisher = {IJCAI Organization},
  pages     = {5278--5284},
  year      = {2019},
  month     = {7}

}

@inproceedings{zhu2022balanced,
  title={Balanced contrastive learning for long-tailed visual recognition},
  author={Zhu, Jianggang and Wang, Zheng and Chen, Jingjing and Chen, Yi-Ping Phoebe and Jiang, Yu-Gang},
  booktitle={Proceedings of the IEEE/CVF Conference on Computer Vision and Pattern Recognition},
publisher= {IEEE},
  pages={6908--6917},
  year={2022}
}

@inproceedings{bu2024automatic,
  title={Automatic fusion for multimodal entity alignment: A new perspective from automatic architecture search},
  author={Bu, Chenyang and Hong, Yunpeng and Zang, Shiji and Chang, Guojie and Wu, Xindong},
  booktitle={Proceedings of IEEE International Conference on Multimedia and Expo},
  pages={1--6},
  year={2024},
  organization={IEEE}
}

@inproceedings{feng-etal-2022-language,
    title = "Language-agnostic bert sentence embedding",
    author = "Feng, Fangxiaoyu  and
      Yang, Yinfei  and
      Cer, Daniel  and
      Arivazhagan, Naveen  and
      Wang, Wei",
    booktitle = "Proceedings of the Association for Computational Linguistics",
    year = "2022",
    publisher = "ACL",
    pages = "878--891",
}

@inproceedings{devlin-etal-2019-bert,
    title = "{B}ert: Pre-training of deep bidirectional transformers for language understanding",
    author = "Devlin, Jacob  and
      Chang, Ming-Wei  and
      Lee, Kenton  and
      Toutanova, Kristina",
    booktitle = "Proceedings of Conference of the North {A}merican Chapter of the Association for Computational Linguistics: Human Language Technologies",
    year = "2019",
    publisher = "ACL",
    pages = "4171--4186",
}

@article{kasneci2023chatgpt,
  title={ChatGPT for good? On opportunities and challenges of large language models for education},
  author={Kasneci, Enkelejda and Se{\ss}ler, Kathrin and K{\"u}chemann, Stefan and Bannert, Maria and Dementieva, Daryna and Fischer, Frank and Gasser, Urs and Groh, Georg and G{\"u}nnemann, Stephan and H{\"u}llermeier, Eyke and others},
  journal={Learning and individual differences},
  volume={103},
  pages={102274},
  year={2023},
  publisher={Elsevier}
}

@inproceedings{wang2024pseudo,
  title={Pseudo-label calibration semi-supervised multi-modal entity alignment},
  author={Wang, Luyao and Qi, Pengnian and Bao, Xigang and Zhou, Chunlai and Qin, Biao},
  booktitle={Proceedings of the AAAI conference on artificial intelligence},
  volume={38},
  number={8},
publisher={AAAI},
  pages={9116--9124},
  year={2024}
}

@article{edge2024local,
  title={From local to global: A graph rag approach to query-focused summarization},
  author={Edge, Darren and Trinh, Ha and Cheng, Newman and Bradley, Joshua and Chao, Alex and Mody, Apurva and Truitt, Steven and Larson, Jonathan},
  journal={arXiv preprint arXiv:2404.16130},
  year={2024},
url = {https://arxiv.org/pdf/2404.16130}
}

@article{wu2021knowledge,
  title={Knowledge Graph for China's Genealogy},
  author={Wu, Xindong and Jiang, Tingting and Zhu, Yi and Bu, Chenyang},
  journal={IEEE Transactions on Knowledge and Data Engineering},
  volume={35},
  number={1},
  pages={634--646},
  year={2021},
  publisher={IEEE}
}

@inproceedings{DBLP:conf/ijcai/ZhuoPWW0LW025,
  author       = {Xingrui Zhuo and
                  Shirui Pan and
                  Jiapu Wang and
                  Gongqing Wu and
                  Zan Zhang and
                  Rui Li and
                  Zizhong Wei and
                  Xindong Wu},
  title        = {Progressive Prefix-Memory Tuning for Complex Logical Query Answering
                  on Knowledge Graphs},
  booktitle    = {Proceedings of the Thirty-Fourth International Joint Conference on
                  Artificial Intelligence, {IJCAI} 2025, Montreal, Canada, August 16-22,
                  2025},
  pages        = {3716--3724},
  publisher    = {ijcai.org},
  year         = {2025}
}

@inproceedings{zeng2023imgcl,
  title={Imgcl: Revisiting graph contrastive learning on imbalanced node classification},
  author={Zeng, Liang and Li, Lanqing and Gao, Ziqi and Zhao, Peilin and Li, Jian},
  booktitle={Proceedings of the AAAI Conference on Artificial Intelligence},
  volume={37},
  number={9},
publisher = {AAAI Press},
  pages={11138--11146},
  year={2023}
}

@inproceedings{cui2021parametric,
  title={Parametric contrastive learning},
  author={Cui, Jiequan and Zhong, Zhisheng and Liu, Shu and Yu, Bei and Jia, Jiaya},
  booktitle={Proceedings of the IEEE/CVF International Conference on Computer Vision},
  pages={715--724},
publisher = "IEEE",
  year={2021}
}

@inproceedings{kang2020exploring,
  title={Exploring balanced feature spaces for representation learning},
  author={Kang, Bingyi and Li, Yu and Xie, Sa and Yuan, Zehuan and Feng, Jiashi},
  booktitle={Proceedings of International conference on Learning Representations},
  publisher    = {OpenReview.net},
  year         = {2021},
  url          = {https://openreview.net/forum?id=OqtLIabPTit},
}

@article{chen2024entity,
  title={Entity alignment with noisy annotations from large language models},
  author={Chen, Shengyuan and Zhang, Qinggang and Dong, Junnan and Hua, Wen and Li, Qing and Huang, Xiao},
  journal={arXiv preprint arXiv:2405.16806},
  year={2024},
  URL = {https://arxiv.org/pdf/2405.16806}
}

@article{jiang2023integrating,
  title={Integrating symbol similarities with knowledge graph embedding for entity alignment: An unsupervised framework},
  author={Jiang, Tingting and Bu, Chenyang and Zhu, Yi and Wu, Xindong},
  journal={Intelligent Computing},
  volume={2},
  pages={0021},
  year={2023},
  publisher={AAAS}
}

@article{tang2023fused,
  title={A fused gromov-wasserstein framework for unsupervised knowledge graph entity alignment},
  author={Tang, Jianheng and Zhao, Kangfei and Li, Jia},
  journal={arXiv preprint arXiv:2305.06574},
  year={2023},
url ={https://arxiv.org/pdf/2305.06574}
}

@inproceedings{zeng2021towards,
  title={Towards entity alignment in the open world: An unsupervised approach},
  author={Zeng, Weixin and Zhao, Xiang and Tang, Jiuyang and Li, Xinyi and Luo, Minnan and Zheng, Qinghua},
  booktitle={Proceedings of Database Systems for Advanced Applications},
  pages={272--289},
  year={2021},
  organization={Springer}
}

@inproceedings{he2019unsupervised,
  title={Unsupervised entity alignment using attribute triples and relation triples},
  author={He, Fuzhen and Li, Zhixu and Qiang, Yang and Liu, An and Liu, Guanfeng and Zhao, Pengpeng and Zhao, Lei and Zhang, Min and Chen, Zhigang},
  booktitle={Proceedings of Database Systems for Advanced Applications},
  pages={367--382},
  year={2019},
  organization={Springer}
}

@inproceedings{guo2019learning,
  title={Learning to exploit long-term relational dependencies in knowledge graphs},
  author={Guo, Lingbing and Sun, Zequn and Hu, Wei},
  booktitle={Proceedings of International Conference on Machine Learning},
  pages={2505--2514},
  year={2019},
  organization={PMLR}
}

@inproceedings{yang2020cotsae,
  title={Cotsae: Co-training of structure and attribute embeddings for entity alignment},
  author={Yang, Kai and Liu, Shaoqin and Zhao, Junfeng and Wang, Yasha and Xie, Bing},
  booktitle={Proceedings of the AAAI Conference on Artificial Intelligence},
  volume={34},
  number={03},
publisher = {IJCAI Organization},
  pages={3025--3032},
  year={2020}
}

@inproceedings{mclea,
   author = {Zhenxi Lin and Ziheng Zhang and Meng Wang and Yinghui Shi and Xian Wu and Yefeng Zheng},
   booktitle = {Proceedings of International Conference on Computational Linguistics},
   pages = {2572-2584},
   title = {Multi-modal contrastive representation learning for entity alignment},
   year = {2022},
    publisher={ICCL}
}

@inproceedings{ni2023psnea,
  title={Psnea: Pseudo-siamese network for entity alignment between multi-modal knowledge graphs},
  author={Ni, Wenxin and Xu, Qianqian and Jiang, Yangbangyan and Cao, Zongsheng and Cao, Xiaochun and Huang, Qingming},
  booktitle={Proceedings of International Conference on Multimedia},
  pages={3489--3497},
publisher = {ACM},
  year={2023}
}

@article{meng2025se,
  title={SE-GNN: Seed Expanded-Aware Graph Neural Network with Iterative Optimization for Semi-supervised Entity Alignment},
  author={Meng, Tao and Shan, Shuo and Shao, Hongen and Shou, Yuntao and Ai, Wei and Li, Keqin},
  journal={IEEE Transactions on Knowledge and Data Engineering},
  year={2025},
  publisher={IEEE}
}

@inproceedings{iclea,
author = {Zeng, Kaisheng and Dong, Zhenhao and Hou, Lei and Cao, Yixin and Hu, Minghao and Yu, Jifan and Lv, Xin and Cao, Lei and Wang, Xin and Liu, Haozhuang and Huang, Yi and Feng, Junlan and Wan, Jing and Li, Juanzi and Feng, Ling},
title = {Interactive Contrastive Learning for Self-Supervised Entity Alignment},
year = {2022},
publisher = {ACM},
booktitle = {Proceedings of ACM International Conference on Information \& Knowledge Management},
pages = {2465–2475},
numpages = {11},

}

@inproceedings{zhuo2025effective,
  title={Effective Instruction Parsing Plugin for Complex Logical Query Answering on Knowledge Graphs},
  author={Zhuo, Xingrui and Wang, Jiapu and Wu, Gongqing and Pan, Shirui and Wu, Xindong},
  booktitle={Proceedings of the ACM on Web Conference},
  pages={4780--4792},
PUBLISHER={ACM},
  year={2025}
}

@article{yu2020cross,
  title={Cross-modal knowledge reasoning for knowledge-based visual question answering},
  author={Yu, Jing and Zhu, Zihao and Wang, Yujing and Zhang, Weifeng and Hu, Yue and Tan, Jianlong},
  journal={Pattern Recognition},
  volume={108},
  pages={107563},
  year={2020},
  publisher={Elsevier}
}

@inproceedings{retrieval,
  title={Deep supervised cross-modal retrieval},
  author={Zhen, Liangli and Hu, Peng and Wang, Xu and Peng, Dezhong},
  booktitle={Proceedings of the IEEE/CVF Conference on Computer Vision and Pattern Recognition},
PUBLISHER= {IEEE},
  pages={10394--10403},
  year={2019}
}

@inproceedings{gat,
  title={Graph attention networks},
  author={Veli{\v{c}}kovi{\'c}, Petar and Cucurull, Guillem and Casanova, Arantxa and Romero, Adriana and Lio, Pietro and Bengio, Yoshua},
  booktitle={Proceedings of International Conference on Learning Representations},
publisher={Openreview.net},
  year={2018}
}

@article{ektefaie2023multimodal,
  title={Multimodal learning with graphs},
  author={Ektefaie, Yasha and Dasoulas, George and Noori, Ayush and Farhat, Maha and Zitnik, Marinka},
  journal={Nature Machine Intelligence},
  volume={5},
  number={4},
  pages={340--350},
  year={2023},
  publisher={Nature Publishing Group UK London}
}

@ArtifactSoftware{R,
    title = {R: A Language and Environment for Statistical Computing},
    author = {{R Core Team}},
    organization = {R Foundation for Statistical Computing},
    address = {Vienna, Austria},
    year = {2019},
    url = {https://www.R-project.org/},
}

% \clearpage
\appendix

% Our appendix consists of three sections. The first section provides additional insights into the theoretical analysis, the second section offers supplementary algorithms related to the methodology, and the third section includes further details about the experimental section.
% In the supplementary experiments, we incorporated a comparison with the latest algorithm PCMEA, presented the results of PSQE on two benchmark datasets, ICEWS\_WIKI and ICEWS\_YAGO, and provided further analysis of several additional ablation studies.

\section{Complementary Theoretical Analysis}
% The theoretical analysis is complemented by three parts, the first of which is an introduction to existing pseudo seed generation methods, the second is a derivation of Theorem 1, and the third is an introduction to the gradient propagation process.

\subsection{Pseudo-Alignment Seed Generation}

Existing pseudo-alignment seed generation methods ~\cite{mclea,chen2023meaformer} are mainly based on the Unsupervised Visual Pivoting (UVP) proposed in EVA, which computes visual similarity to infer the correspondence between knowledge graphs. Specifically, UVP first computes the visual cosine similarity \( \text{sim}(e_1^{i}, e_2^{j}) \) between all cross graph entities $( e_1^{i}, e_2^{j})$ in the dataset, i.e:
\begin{equation}
\footnotesize
    \text{sim}(e_1^{i}, e_2^{j}) = \frac{h^i_1 \cdot h^j_2}{\|h^i_1\| \|h^j_2\|}
\end{equation}
where \( v(e_1^{i}) \) and \( v(e_2^{j}) \) denote the visual representations of the entities \(e_1^{i} \) and \(e_2^{j} \), respectively, and \(\|v(e)\|\) is its paradigm. UVP then constructs a cosine similarity matrix $ \text{SimMatrix}$  to capture the similarity between all pairs of entities \( (e_1^{i}, e_2^{j}) \):
\begin{equation}
\footnotesize
    \text{SimMatrix} = \left[ \text{sim}(e_1^{i}, e_2^{j}) \right]_{i,j}
\end{equation}

Next, the entity pairs in the matrix are sorted by their similarity from highest to lowest to obtain a sorted list of entity pairs \( \text{SortedPairs} \). Based on this, we sequentially select the most similar entity pairs from the sorted entity pairs as visual pivots until we collect \( k \) pairs of the most similar entity pairs, denoted as \( \{(e_1^{i_1}, e_2^{j_1}), (e_1^{i_2}, e_2^{j_2}), \dots, (e_1^{i_k}, e_2^{j_k} )\} \). In this process, whenever a pair of entities \( (e_1^{i_m}, e_2^{j_m})\) is selected, we exclude all other links related to these two entities.

The result is  the pseudo-alignment
seed $\mathcal{S}$ containing $k$ pairs of visually most similar pairs of entities, where both entities in each pair are unduplicated, in the following form:
\begin{equation}
\footnotesize
    \mathcal{S} = \left\{ (e_1^{i_1}, e_2^{j_1}), (e_1^{i_2}, e_2^{j_2}), \dots, (e_1^{i_k}, e_2^{j_k}) \right\}, \quad e_1^{i} \neq e_2^{j} \quad \forall i, j
\end{equation}

\subsection{Proof of Theorem 1}
\label{proof1}
Inspired by the formulation of contrastive learning analysis in~\cite{zhu2022balanced,khosla2020supervised}, we first simplify Eq.~\ref{icl} for clarity by omitting $\tau$ and making an approximate calculation. Substituting Eq.~\ref{pm} into Eq.~\ref{icl}, we obtain:
\begin{equation}
\footnotesize
\vspace{-1em}
\begin{aligned}
L^{ICL}\approx & -\sum_{i\in B}\log(\frac{1}{2}(\frac{\exp{(h_{1}^{i}}^T\cdot h_{2}^{i})}{\exp(h_{1}^{i}\cdot h_{2}^{i})+\sum\limits_{j\in N_{i}^{ig}}\exp({h_{1}^{i}}^T\cdot h^{j})} \\
&+\frac{\exp({h_{2}^{i}}^T\cdot h_{1}^{i})}{\exp({h_{2}^{i}}^T\cdot h_{1}^{i})+\sum\limits_{j\in N_{i}^{ig}}\exp({h_{2}^{i}}^T\cdot h^{j})})) 
\end{aligned}
\label{eq12}
\end{equation}
where $B$ is the pseudo seed data representing a batch.

When we consider only one-way alignment, Eq.~\ref{eq12} can be simplified as:
\begin{equation}
\footnotesize
    L^{ICL} \approx  - \sum\limits_{i \in B} {\log (\frac{{\exp ({h_1^i}^T \cdot h_2^i)}}{{\exp ({h_1^i}^T \cdot h_2^i) + \sum\limits_{j \in {\mathcal N}_i^{ng}} {\exp ({h_1^i}^T \cdot {h^j})} }}} )
\end{equation}
Since the exponential function is convex, by applying Jensen's inequality:
\begin{equation}
\footnotesize
\begin{aligned}
    L^{ICL} &\ge \sum\limits_{i \in B} {\log (1 + \frac{{(|B| - 1)\exp (\frac{1}{{|B| - 1}}\sum\limits_{j \in {\mathcal N}_i^{ng}} {{h_1^i}^T \cdot  {h^j}}) }}{{\exp (h_1^i \cdot h_2^i)}}} ) \\
    &  = \sum\limits_{i \in B} {\log (1 + } (|B| - 1)\exp ({\frac{1}{{(|B| - 1)}}\sum\limits_{j \in {\mathcal N}_i^{ng}} {{h_1^i}^T\cdot {h^j}}}  -  {{h_1^i}^T \cdot h_2^i}))
\end{aligned}
\label{equ14}
\end{equation}
All inner products ${h_1^i}^T  \cdot h^j = \sigma$ hold if and only if there exist values $\sigma$ such that $\forall {j \in {\mathcal N}_i^{ng}}$. 

Since the embedding vectors satisfy $h=1$, then $||h^i-h^j||=1-2h^{iT}\cdot h^j$, thus Eq.~\ref{equ14} can be expressed as:
 \begin{equation}
 \footnotesize
{L^{ICL}} \ge \sum\limits_{i \in B} {\log (1 + } D\exp (\underbrace {\frac{1}{2}||h_1^i - h_2^i||}_{attraction} - \underbrace {\frac{1}{{2D}}\sum\limits_{j \in {\mathcal N}_i^{ng}} {||h_1^i - {h^j}||} }_{repulsion} - \frac{1}{D}))
 \end{equation}
 where ${\mathcal N}_i^{ng}$ denotes the set of negative samples for entity i, and $D=|B|-1$ corresponds to the number of negative instances in the mini-batch, excluding the anchor entity itself.

\subsection{Gradient Analysis}
\label{proof2}
Building upon the formulation of contrastive learning analysis in~\cite{zhu2022balanced,khosla2020supervised}, the ICL loss in Eq.~\ref{icl} is rewritten by neglecting the temperature coefficient $\tau$. Its gradient is given as follows:
\begin{equation}
\footnotesize
\begin{array}{l}
\frac{{\partial {L^{ICL}}}}{{\partial h_1^i}} =  - \frac{\partial }{{\partial h_1^i}}(\sum\limits_{i \in B} {\log (\frac{{\exp ({h_1^i}^T  \cdot h_2^i)}}{{\exp ({h_1^i}^T  \cdot h_2^i) + \sum\limits_{j \in {\mathcal N}_i^{ng}} {\exp ({h_1^i}^T  \cdot {h^j})} }}} ))\\
 =  - \sum\limits_{i \in B} {(\frac{{\partial {h_1^i}^T  \cdot h_2^i}}{{\partial h_1^i}}}  - \frac{\partial }{{\partial h_1^i}}(\log (\exp ({h_1^i}^T  \cdot h_2^i) + \sum\limits_{j \in {\mathcal N}_i^{ng}} {\exp ({h_1^i}^T  \cdot {h^j})} ))\\
 =  - \sum\limits_{i \in B} {(h_2^i}  - \frac{{h_2^i \cdot \exp ({h_1^i}^T  \cdot h_2^i) + \sum\limits_{j \in {\mathcal N}_i^{ng}} {{h^j} \cdot \exp ({h_1^i}^T  \cdot {h^j})} }}{{\exp ({h_1^i}^T  \cdot h_2^i) + \sum\limits_{j \in {\mathcal N}_i^{ng}} {\exp ({h_1^i}^T  \cdot {h^j})} }})\\
 = \sum\limits_{i \in B} {h_2^i({P_{ii}} - 1) + } \sum\limits_{j \in {\mathcal N}_i^{ng}} {\sum\limits_{i \in B} {{h^j}{P_{ij}}} } \\
 = \sum\limits_{i \in B} {h_2^i({P_{ii}} - 1) + } B\sum\limits_{j \in {\mathcal N}_i^{ng}} {{h^j}{P_{ij}}} 
\end{array}
\end{equation}
Where $B$ is the pseudo seed data representing a batch, and we have defined:
\begin{equation}
    \begin{array}{l}
{P_{ii}} = \frac{{\exp ({h_1^i}^T  \cdot h_2^i)}}{{\exp ({h_1^i}^T  \cdot h_2^i) + \sum\limits_{j \in {\mathcal N}_i^{ng}} {\exp ({h_1^i}^T  \cdot {h^j})} }}\\
{P_{ij}} = \frac{{\exp ({h_1^i}^T  \cdot {h^j})}}{{\exp ({h_1^i}^T  \cdot h_2^i) + \sum\limits_{j \in {\mathcal N}_i^{ng}} {\exp ({h_1^i}^T  \cdot {h^j})} }}
\end{array}
\end{equation}
Since a normalization function is applied before computing the ICL loss, we denote the output before normalization as $w_i$, with a slight abuse of notation, i.e., $h_1^i = w_i/||w_i||$. The gradient concerning $w_i$ is then given by:
\begin{equation}
\footnotesize
\begin{array}{l}
\frac{{\partial {L^{ICL}}}}{{\partial {w_i}}} = \frac{1}{{||{w_i}||}}(I - {h_1^i}^T  \cdot h{_1^i}^T)(\sum\limits_{i \in B} {h_2^i({P_{ii}} - 1) + B} \sum\limits_{j \in {\mathcal N}_i^{ng}} {{h^j}{P_{ij}}} )\\
 = \frac{1}{{||{w_i}||}}(\underbrace {\sum\limits_{i \in B} {(h_2^i - (} {h_1^i}^T  \cdot h_2^i) \cdot h_1^i)({P_{ii}} - 1)}_{attraction} \\
 + \underbrace {B\sum\limits_{j \in {\mathcal N}_i^{ng}} {({h^j} - } ({h_1^i}^T  \cdot {h^j}) \cdot h_1^i){P_{ij}}}_{repulsion})
\end{array}
\end{equation}

For the repulsion term, if the entity is very dissimilar to the negative sample, then ${h_1^i}^T  \cdot {h^j} = 0$. So the gradient of ICL is for the repulsion term:
\begin{equation}
\footnotesize
\begin{array}{l}
\sum\limits_{j \in {\mathcal N}_i^{ng}} {||{h^j} - } ({h_1^i}^T  \cdot {h^j}) \cdot h_1^i|| \cdot |{P_{ij}}| = \sum\limits_{j \in {\mathcal N}_i^{ng}} {|{P_{ij}}|} \\
 = \frac{{|{\mathcal N}_i^{ng}|}}{{\exp ({h_1^i}^T  \cdot h_2^i) + \sum\limits_{j \in {\mathcal N}_i^{ng}} {\exp ({h_1^i}^T  \cdot {h^j})} }}
\end{array}
\end{equation}
Therefore, it can be seen that the terms in the denominator are consistent for all negative class samples, so the negative class gradient is proportional to the number of samples. However, when the graph coverage balance of negative samples is poor, the more aggregated entities in the pseudo-aligned seed will be more than the scattered entities, which leads to the gradient of the more aggregated entities being larger than the scattered entities, and the optimization pays excessive attention to the more aggregated entities, which leads to the imbalance of the feature space.

\section{Supplementary Notes on Experiments}

\subsection{Impact of the Multimodal Error Correction  Mechanism}
 To isolate the standalone impact of MIC, we ablate it across three datasets by removing the MIC module from PSQE. The results are shown in Table~\ref{tab:mic-ablation}.
Removing the MIC module leads to a significant drop in pseudo seed precision—16.3\%, 18.5\%, and 15.9\% on ZH\_EN, JA\_EN, and FR\_EN, respectively. This demonstrates the indispensable role of MIC in improving seed quality. For example, on the FR\_EN dataset, MIC filters out noisy pairs generated during neighborhood expansion, such as the incorrect alignment (1841, 35198) from the candidate pool of (25047, 36136).
Moreover, we further assessed the impact of MIC on downstream MMEA performance. As shown in Fig.~\ref{RQ22} of the main paper, using PSQE-generated seeds without MIC in the EVA model leads to a noticeable performance drop, i.e., H\@1 on ZH\_EN decreases from 93.3\% to 92.4\%. These results attest that MIC enhances the quality of pseudo seeds and can benefit various MMEA models in general.

\begin{table}[htbp]
\centering
\caption{Effect of the MIC module on pseudo seed precision.}
\begin{tabular}{lccc}
\toprule
\textbf{Model} & \textbf{ZH\_EN} & \textbf{JA\_EN} & \textbf{FR\_EN} \\
\midrule
PSQE         & 82.5 & 81.0 & 81.0 \\
PSQE - MIC   & 66.2 & 62.5 & 65.1 \\
\bottomrule
\end{tabular}
\label{tab:mic-ablation}
\vspace{-0.5em}
\end{table}

\subsection{More Discussion on Time Consumption}
Table~\ref{timecost} demonstrates the time consumption of each stage of PSQE on the three datasets of DBP15K. Results show that while Stage II involves additional computation due to all-pair similarity calculations, the overall cost remains acceptable. For instance, using MEAformer as the alignment model on an A100 GPU, the full training process takes approximately 51.8 minutes, while our complete pseudo seed generation process across all three stages takes only about 10 minutes, accounting for roughly 20\% of the total runtime. Results show that although Stage II incurs extra cost, it remains practically efficient relative to the overall training time. 

\begin{table}[h]
\caption{Efficiency analysis of PSQE. PSQE runtime (seconds) statistics for generating pseudo-aligned seeds on the three DBP15K datasets.}
\begin{tabular}{ccllcllcll}
\hline
 & \multicolumn{3}{c}{ZH-EN} & \multicolumn{3}{c}{JA-EN} & \multicolumn{3}{c}{FR-EN} \\ \hline
stage \uppercase\expandafter{\romannumeral 1} & \multicolumn{3}{c}{75.3} & \multicolumn{3}{c}{64.3} & \multicolumn{3}{c}{66.2} \\
stage \uppercase\expandafter{\romannumeral 2} & \multicolumn{3}{c}{502.4} & \multicolumn{3}{c}{457.5} & \multicolumn{3}{c}{458.6} \\
stage \uppercase\expandafter{\romannumeral 3} & \multicolumn{3}{c}{5.2} & \multicolumn{3}{c}{4.7} & \multicolumn{3}{c}{5.4} \\ \hline
\end{tabular}
\label{timecost}
\vspace{-0.5em}
\end{table}

% From a theoretical perspective, (1) Stage I has a complexity of $O(N^{2}·d)$, primarily due to multi-modal similarity computation and clustering, where $N$ is the number of entities and $d$ is the embedding dimension; (2) Stage II also incurs $O(N^{2}·d)$ complexity, driven by the similarity matrix construction and global correction mechanism. Although it is the most compute-intensive step, it is performed only once before alignment training, and thus remains scalable; (3) Stage III operates at $O(S^{2}·d)$ complexity, with $S$ being the number of pseudo seeds. It includes neighborhood expansion ($\approx O(S \cdot K \cdot d) $, where $K$ is the average neighbor size) and a secondary correction step.

\begin{table}[htbp]
\centering
\caption{Comparison of H@1 accuracy with and without PSQE under fixed pseudo seed sizes.}
\begin{tabular}{lccc}
\toprule
\textbf{Method} & \textbf{ZH\_EN} & \textbf{JA\_EN} & \textbf{FR\_EN} \\
\midrule
MEAformer        & 0.838 & 0.887 & 0.929 \\
MEAformer+PSQE   & 0.842 & 0.892 & 0.932 \\
\bottomrule
\end{tabular}
\label{tab:psqe-fixed}
\vspace{-1em}
\end{table}

\vspace{-0.5em}
\subsection{More Discussion on Pseudo-seed}
To evaluate the adaptability and robustness of PSQE, we conduct two groups of experiments. 
Unlike methods that require manual specification of the number of pseudo seeds, PSQE automatically adjusts the size based on data distribution and the quality of initial seeds. For example, with an initial 1000 pseudo seeds on the ZH\_EN, JA\_EN, and FR\_EN datasets, PSQE expands the seed set to 2024, 2171, and 2007, respectively. Conversely, on the DWY15K-V1 and DWY15K-V2 datasets, it prunes the initial 5000 seeds down to 4325 and 4134. This demonstrates PSQE's flexibility in either expanding or reducing the pseudo seed set as needed.
We further evaluate PSQE under fixed pseudo seed settings to validate its robustness. We compare the baseline MEAformer with MEAformer+PSQE on ZH\_EN, JA\_EN, and FR\_EN datasets. As shown in Tab.~\ref{tab:psqe-fixed}, PSQE consistently improves alignment performance even without changing the number of pseudo seeds, owing to joint optimization of seed precision and distributional coverage.

% \subsection{More Discussion on Missing Modalities}
% \begin{table}[h]
% \centering
% \caption{Performance comparison of imputation strategies on the FR\_EN dataset.}
% \begin{tabular}{lccc}
% \toprule
% \textbf{Method} & \textbf{H@1} & \textbf{H@10} & \textbf{MRR} \\
% \midrule
% PSQE+Gauss & 0.932 & 0.986 & 0.953 \\
% PSQE+Zero  & 0.928 & 0.983 & 0.948 \\
% \bottomrule
% \end{tabular}
% \label{tab:imputation}
% \end{table}
% To investigate the impact of modality imbalance in MMEA, we conduct a controlled experiment comparing two commonly used imputation strategies for handling missing modality information in the visual space: (1) Zero-vector imputation, where the missing features are filled with zero vectors; and (2) Gaussian-based imputation, where the missing embeddings are sampled from a Gaussian distribution estimated using the empirical mean and standard deviation of the observed visual embeddings.

% In Table~\ref{tab:imputation}, we implement both strategies on the FR\_EN dataset using the MEAformer framework, and evaluate their performance using standard metrics including $H@1$, $H@10$, and $MRR$. 
% The Gaussian-based imputation method achieves higher performance across all metrics, with an improvement of +0.4\% in H@1 compared to zero imputation. This demonstrates its advantage in generating more realistic and distributed embeddings, thereby reducing the concentration of missing entities in the shared space and improving alignment quality.

\end{sloppypar} 

\end{document}